\documentclass{ecai}

\usepackage{graphicx}
\usepackage{latexsym}
\usepackage{amsfonts}
\usepackage{amsthm}
\usepackage{color}
\usepackage{amsmath}
\usepackage[linesnumbered,ruled]{algorithm2e}

\usepackage{tabularx}
\usepackage{stfloats}
\usepackage{multirow}
\usepackage{threeparttable}
\usepackage{makecell}
\usepackage{bm}
\usepackage{nicefrac}
\usepackage{booktabs}
\usepackage{array}

\newtheorem{definition}{Definition}
\newtheorem{lemma}{Lemma}

\usepackage{times}
\usepackage{soul}
\usepackage{url}
\usepackage[hidelinks]{hyperref}
\usepackage[utf8]{inputenc}
\usepackage[small]{caption}
\usepackage{booktabs}
\usepackage{amssymb,tikz,multirow}

\newtheorem{thm}{Theorem}
\newtheorem{example}{Example}
\newtheorem{remark}{Remark}
\newtheorem{corollary}{Corollary}

\usepackage{amssymb, latexsym, amsfonts,blkarray,color,bbm,cancel,tikz,multirow}
\tikzset{edge/.style = {->,> = latex'}}
\usetikzlibrary{arrows}
\SetKwInput{KwInput}{Input}             
\SetKwInput{KwOutput}{Output}

\begin{document}

\begin{frontmatter}

\title{Minimum Target Sets in Non-Progressive Threshold Models: When Timing Matters}

\author[A]{\fnms{Hossein}~\snm{Soltani}}
\author[B]{\fnms{Ahad}~\snm{N. Zehmakan}\thanks{Corresponding Author. Email: ahadn.zehmakan@anu.edu.au.}}
\author[A]{\fnms{Ataabak}~\snm{B. Hushmandi}} 

\address[A]{Urmia University of Technology}
\address[B]{The Australian National University}

\begin{abstract}
Let $G$ be a graph, which represents a social network, and suppose each node $v$ has a threshold value $\tau(v)$. Consider an initial configuration, where each node is either positive or negative. In each discrete time step, a node $v$ becomes/remains positive if at least $\tau(v)$ of its neighbors are positive and negative otherwise. A node set $\mathcal{S}$ is a Target Set (TS) whenever the following holds: if $\mathcal{S}$ is fully positive initially, all nodes in the graph become positive eventually.
We focus on a generalization of TS, called Timed TS (TTS), where it is permitted to assign a positive state to a node at any step of the process, rather than just at the beginning.

We provide graph structures for which the minimum TTS is significantly smaller than the minimum TS, indicating that timing is an essential aspect of successful target selection strategies. Furthermore, we prove tight bounds on the minimum size of a TTS in terms of the number of nodes and maximum degree when the thresholds are assigned based on the majority rule.

We show that the problem of determining the minimum size of a TTS is NP-hard and provide an Integer Linear Programming formulation and a greedy algorithm. We evaluate the performance of our algorithm by conducting experiments on various synthetic and real-world networks. We also present a linear-time exact algorithm for trees.
\end{abstract}

\end{frontmatter}

\section{Introduction}

Over the past few decades, the world has experienced an extreme surge in the proliferation of online social networks. These platforms have emerged as an omnipresent facet of contemporary societies, enabling people to forge connections with their peers, aggregate information, and express their opinions. As such, these networks have become the principal conduits for the rapid dissemination of information, facilitating the fluid formation of opinions through online interactions. Consequently, marketing firms and political campaigns frequently exploit social networks to achieve their desired outcomes, cf.~\cite{kempe}. These entities harness the power of these platforms to advertise new consumer goods and promote political factions. The underlying strategy often revolves around the notion that pinpointing a select group of influential individuals within a given community could trigger a massive ripple effect of influence across the network at large. This has catalyzed a growing interest in the quantitative analysis of opinion diffusion and collective decision-making mechanisms, cf.~\cite{FG,coro2019exploiting,C}.

From a theoretical standpoint, it is pertinent to introduce and investigate mathematical models of influence diffusion, which simulate the process of how individuals revise their opinions and how the influence disseminates through social interactions. The majority of the proposed models utilize a graph, denoted by $G$, to model the interactions between members of a community, cf.~\cite{nichterlein2013tractable}. The graph, intended to represent a social network, features each node as an individual, with an edge representing a relationship between individuals such as friendship, collaboration, or mentorship. Furthermore, each node is typically assigned a binary state, representing a positive or negative stance regarding a specific topic or the status on the adoption of a novel technological product. Thereafter, nodes continue to update their states as a function of their neighboring nodes' states.

One category of models which has gained significant popularity is the class of threshold models, cf.~\cite{nichterlein2013tractable,FG,ackerman}. These models entail that each node $v$ possesses a distinct threshold value $\tau(v)$ and updates its opinion to a positive state only when the number of its positive neighbors exceeds the stipulated threshold. In these models, while peer pressure could potentially influence a node's decision-making process, nodes can exhibit varying degrees of resistance, where those with higher threshold values require a greater number of positive connections to adopt a positive state. 

The majority of threshold based models explored in previous research fall under the umbrella of two categories: progressive and non-progressive, cf.~\cite{C, kempe,zehmakan2023random}. Progressive models are designed to simulate situations where states evolve in a fixed direction, i.e., once a node assumes a positive disposition, it remains positive indefinitely. This type of dynamic is particularly suited to scenarios where nodes transition from an uninformed (negative) to an informed (positive) state, or where a node adopts a new technology (i.e., switches from negative to positive). Conversely, in the non-progressive setting, nodes possess the capability of oscillating between positive and negative states. In this context, a node's state represents its stance on a given topic (such as levying additional taxes on alcoholic beverages), favoring one of two political parties, or embracing one of two competing services.

A set of nodes whose agreement on positive state results in the whole (or a large body of) network eventually adopting a positive stance is called a \textit{target set}, cf.~\cite{C}. In order to acquire insights into the most effective manipulation strategies for controlling the outcome of opinion formation dynamics, the problem of identifying the minimum size of a target set has extensively been examined. This problem is commonly referred to as \emph{target set selection}~\cite{ackerman} or \emph{influence maximization}~\cite{kempe}, depending on its exact formulation, and has yielded a plethora of hardness, algorithmic, and combinatorial findings, cf.~\cite{nichterlein2013tractable,FG,coro2019exploiting,C,kempe}.



In the present work, we introduce a generalization of a target set, where we allow the nodes to be targeted at different steps of the process (rather than all at once). Such a set of targeted nodes is called a \textit{timed target set}. Some prior work has considered the framework where the manipulator intervenes in the process in a more dynamic fashion, but their setup is fundamentally different from ours, in terms of the underlying diffusion process, the permitted intervention operations, and the manipulator's objectives, cf.~\cite{chen2012time,aggarwal2012influential,kim2014ct, zhang2013maximizing}.

We investigate the problem of finding the minimum size of a timed target set in a non-progressive threshold model. We provide some hardness results, propose a greedy algorithm whose performance is evaluated on real and synthetic graph data, present an exact linear-time algorithm for trees, and prove tight bounds on the minimum size of a timed target set in terms of different graph parameters.

\textbf{Outline.} In the rest of this section, we first give some basic definitions in Subsection~\ref{preliminaries}. Building on that, we present our contributions and give an overview of related work in the following two subsections. Our bounds on the minimum size of a timed target set are proven in Section~\ref{section Lower bound}. Our complexity and algorithmic results are given in Section~\ref{algorithms}.

\subsection{Preliminaries}\label{preliminaries}
Consider a simple undirected graph $G=(V,E)$. We let $n:=|V|$, $m:=|E|$ and use the shorthand $vu$ (or $uv$) for an edge $\{v,u\}\in E$. Let $N(v):=\{u\in V: vu\in E\}$ be the \emph{neighborhood} of $v$ and $N[v]=N(v)\cup\{v\}$ be the \emph{closed neighborhood} of $v$. $d(v):=|N(v)|$ denotes the degree of $v$ and $\Delta$ stands for the maximum degree in $G$. Furthermore, for $D\subseteq V$, let $d_D(v):=|N(v)\cap D|$. We say $G$ is an \emph{even} graph if $d(v)$ is even for every node $v$ in $G$. For $D\subseteq V$, the induced subgraph of $G$ on $D$ is denoted by $G[D]$, and $G\setminus D$ stands for the induced subgraph on $V\setminus D$. In case of $D=\{v\}$ for some $v\in V$, we use the notation $G\setminus v$ instead of $G\setminus \{v\}$.

By the \emph{threshold assignment} for the nodes of a graph $G$, we mean a function $\tau:V\longrightarrow \mathbb{N}\cup\{0\}$ such that for each node $v\in V$ the inequality $0\leq\tau(v)\leq d(v)$ holds. Some special choices of the threshold assignment are the \emph{strict majority} $\tau(v)=\left\lceil (d(v)+1)/2\right\rceil$ and \emph{simple majority} $\tau(v)=\left\lceil d(v)/2\right\rceil$.

Consider a pair $(G,\tau)$ and an initial \emph{configuration} where each node is either \emph{positive} or \emph{negative}. In the \emph{progressive threshold model}, in each discrete time step, a negative node $v$ becomes positive if at least $\tau(v)$ of its neighbors are positive and positive nodes remain unchanged. In the \emph{non-progressive threshold model}, a node $v$ becomes positive if at least $\tau(v)$ of its neighbors are positive and becomes negative otherwise. Note that in the non-progressive model, nodes can switch from positive to negative while this is not possible in the progressive model. Furthermore, we define $\mathcal{A}_i$ to be the set of positive nodes in the $i$-th step of the process.

A set $\mathcal{S}\subseteq V$ is called a \emph{target set} (TS) whenever the following holds: If $\mathcal{S}$ is fully positive, then all nodes become positive after some steps, i.e., if $\mathcal{A}_0=\mathcal{S}$, then $\mathcal{A}_i=V$ for some $i\in \mathbb{N}$. For a pair $(G,\tau)$, the minimum size of a TS in the progressive and non-progressive model is denoted by $\overrightarrow{MT}(G,\tau)$ and $\overleftrightarrow{MT}(G,\tau)$, respectively. Note that we use a forward arrow for the progressive model and a bidirectional arrow for the non-progressive model.

According to the definition of a TS, a manipulator targets a set of nodes at once. However, it is sensible to consider the set-up where the manipulator can target nodes at different steps of the process. We capture this by introducing the concept of a \emph{timed target set} (TTS), defined below.

\begin{definition}\label{DTS}
For a pair $(G,\tau)$, the finite sequence $\mathcal{S}_0,\mathcal{S}_1,\ldots,\mathcal{S}_k$ of subsets of $V$, for some integer $k$, is said to be a TTS in the non-progressive model when there is a sequence $Q_0,Q_1,\ldots,Q_k$ of subsets of $V$ such that : (i) $Q_0=\emptyset$; (ii) $v\in Q_i$ for $i\geq1$ if and only if $|N(v)\cap(\mathcal{S}_{i-1}\cup Q_{i-1})|\geq\tau(v)$; (iii) $\mathcal{S}_k=\emptyset$; (iv) $Q_k=V$. We denote the minimum size of a TTS in the non-progressive model with $\overleftrightarrow{MTT}(G,\tau)$, where the size of a TTS is equal to $\sum_{i=0}^{k} |\mathcal{S}_i|$.
\end{definition}

Note that $\mathcal{S}_i\cup Q_i$ is the set of positive nodes in the $i$-th step, i.e., $\mathcal{A}_i$. The nodes of $Q_i$ become positive using their positive neighbors in step $i-1$ (i.e., their neighbors in $\mathcal{S}_{i-1}\cup Q_{i-1}$) but the nodes of $\mathcal{S}_i$ are chosen to be positive in step $i$ by the manipulator.

We do not define TTS for the progressive model since the ability to target nodes at different steps does not give the manipulator any extra power. This is because there is no benefit for a manipulator to target a node in some step $i$, for $i\ge1$, instead of targeting it in step $0$.

As a warm-up, let us provide the example below. 
\begin{example}\label{example}
	Consider the complete bipartite graph $K_{1,n-1}$ with partite sets $X=\{x\}$ and $Y=\{y_1,y_2,\ldots,y_{n-1}\}$. Let $\tau$ be the strict majority threshold, i.e., $\tau(v)=\left\lceil (d(v)+1)/2\right\rceil$ for every node $v$.
	\begin{itemize}
		\item $\mathcal{S}=X$ is a minimum size TS in the progressive model and thus $\overrightarrow{MT}(K_{1,n-1}, \tau)=1$.
		
	\item Let $\mathcal{S}=X\cup Y'$ where $Y'$ is any subset of $Y$ of cardinality $\tau(x)$. It is easy to see that $\mathcal{S}$ is a minimum size TS in the non-progressive model, which yields $\overleftrightarrow{MT}(K_{1,n-1}, \tau)=\tau(x)+1=\lceil\frac{n}{2}\rceil+1$.
		
		\item The sequence $\mathcal{S}_0,\mathcal{S}_1,\mathcal{S}_2$ with $\mathcal{S}_0=\mathcal{S}_1=\{x\}$ and $\mathcal{S}_2=\emptyset$ is a TTS of size $2$ and there is no TTS of size $1$. Hence, $\overleftrightarrow{MTT}(K_{1,n-1},\tau)=2$. 
	\end{itemize}
\end{example}

\subsection{Our Contribution} \label{our_contribution}

We focus on the minimum size of a TTS in the non-progressive threshold model, i.e., $\overleftrightarrow{MTT}(G,\tau)$. We present tight bounds, prove hardness results, and provide approximation and exact algorithms for general and special classes of graphs.

\textbf{Timing Matters.} In the non-progressive model, TTS has two advantages over the original TS: (i) the nodes can be targeted at different steps (ii) a node can be targeted more than once. Example~\ref{example} demonstrates that these two advantages amplify the power of a manipulator significantly since $\overleftrightarrow{MTT}(K_{1,n-1},\tau)$ is much smaller than $\overleftrightarrow{MT}(K_{1,n-1}, \tau)$. What if we require that $\mathcal{S}_i\cap \mathcal{S}_j = \emptyset$ for every two distinct sets $\mathcal{S}_i$ and $\mathcal{S}_j$ in the definition of a TTS, i.e., take away the advantage (ii)? We prove that the advantage (i) suffices to make the manipulator substantially stronger for some classes of graphs and threshold assignments. Thus, targeting nodes at appropriate time is very crucial for successful manipulation and this paper is the first to consider this fundamental aspect of the target set selection.

\textbf{Tight Bounds on $\overleftrightarrow{MTT}(G,\tau)$ for Strict Majority.} We prove that $\overleftrightarrow{MTT}(G,\tau)\ge 2n/(\Delta+1)$ when $\tau$ is the strict majority and this bound is the best possible. (This extends a result from~\cite{FG}.) We first prove this bound for bipartite graphs using some combinatorial and potential function arguments and then extend to the general case. This result implies that for bounded-degree graphs, a na\"ive algorithm which simply targets all nodes has a constant approximation ratio. When the graph $G$ is even, we provide the stronger bound of $4n/(\Delta+2$). The improvement might seem negligible at the first glance, but it is actually quite impactful. For example, for a cycle $C_n$, the first bound is equal to $2n/3$, but the second one gives the tight bound of $n$. Furthermore, while requiring the graph to be even might seem very demanding, it actually captures some important graph classes such as the $d$-dimensional torus. Determining $\overleftrightarrow{MT}(G,\tau)$ and $\overrightarrow{MT}(G,\tau)$ for the $d$-dimensional torus was studied extensively by prior work, due to certain applications in statistical physics, and the exact answer was proven only after a long line of papers, cf.~\cite{morris2009minimal}.

\textbf{Inapproximability Result.} We prove that the problem of finding $\overleftrightarrow{MTT}(G,\tau)$ for a given pair $(G,\tau)$ cannot be approximated within a ratio of $\mathcal{O}(2^{\log^{1-\epsilon}n})$ for any fixed $\epsilon>0$ unless $NP\subseteq DTIME(n^{polylog(n)})$, by a polynomial time reduction from the progressive variant. 

\textbf{Integer Linear Programming Formulation.} A standard approach to tackle an NP-hard problem is to formulate it as an Integer Linear Program (ILP). Then, we can use standard and powerful ILP solvers to solve small-size problems.
We provide an ILP formulation for the problem of finding $\overleftrightarrow{MTT}(G,\tau)$.

\textbf{Greedy Algorithm and Experiments.} We propose a greedy algorithm which finds a TTS for a given pair $(G,\tau)$ and prove its correctness. Then, we provide the outcome of our experiments on various synthetic and real-world graph data. Our experiments on small synthetic networks demonstrate that the minimum size of a TTS is strictly smaller than the minimum size of a TS in most cases, confirming that the effect of timing is not restricted to tailored graphs (such as the one given in Example~\ref{example}). To find the optimal solutions, we rely on our ILP formulation. Furthermore, we observe that our greedy approach returns almost optimal solutions in most cases. We also compare the outcome of our greedy algorithm against an analogous greedy algorithm for TS on real-world networks, such as Facebook and Twitter, and observe a $13\%$ improvement.

\textbf{Exact Linear-Time Algorithm for Trees.} It is known, cf.~\cite{C,Centeno,Ben-Zwi}, that if we limit ourselves to trees, then the problem of determining the minimum size of a TS in the progressive model (i.e., $\overrightarrow{MT}(G,\tau)$) is no longer NP-hard and an exact polynomial time algorithm exists. It was left as an open problem in~\cite{Ben-Zwi} whether a similar result could be proven for the non-progressive variants. Recently, it was shown~\cite{BER} that if $\tau(v)\in\{0,1,d(v)\}$ for every node $v$, then the problem for $\overleftrightarrow{MT}(G,\tau)$ is tractable on trees. We make progress on this front by providing a linear-time algorithm which outputs $\overleftrightarrow{MTT}(G,\tau$) when $G$ is a tree and for any choice of $\tau$. Our algorithm can be interpreted as a dynamic programming approach. It is worth to emphasize that algorithms for trees are not only theoretically interesting, but also the pathway to fixed-parameter tractable algorithms in terms of treewidth, cf.~\cite{Ben-Zwi}, which are relevant from a practical perspective too. 

\textbf{Proof Techniques.} Whilst for some of our results such as hardness proof and greedy algorithm, we leverage the rich literature on TS, for others we need to develop several novel proof techniques. In particular, we introduce several new combinatorial and graph tools for the proof of our bounds. Furthermore, we devise novel techniques to establish a linear-time algorithm for trees, which in fact might be useful to settle the problem for TS (since it is only resolved for a very constrained setup as mentioned above.) 

\subsection{Related Work}\label{prior_work}

Numerous models have been developed and studied to gain more insights into the mechanisms and general principles driving the opinion formation and information spreading among the members of a community, cf.~\cite{faliszewskiopinion,castiglioni2020election,brill2016pairwise,anagnostopoulosbiased,bredereck2017manipulating}. In the plethora of opinion diffusion models, the threshold models have received a substantial amount of attention. While both the progressive and non-progressive threshold models had been studied in the earlier work, cf.~\cite{P,aizenman1988metastability}, they were popularized by the seminal work of Kempe, Kleinberg, and Tardos~\cite{kempe}.

\textbf{Convergence Properties.} In the progressive threshold model, it is straightforward to observe that the process reaches a fixed configuration (where no node can update) in at most $n$ steps. For the non-progressive variant, Goles and Olivos~\cite{GOLES1980187} proved that the process always reaches a cycle of configurations of length one or two (i.e., a fixed configuration or switching between two configurations). Furthermore, this happens in $\mathcal{O}(m)$ steps, where $m$ is the number of edges, according to~\cite{poljak1986pre}. Stronger bounds are known for special cases. For example, a logarithmic upper bound is proven in~\cite{zehmakan2020opinion} for graphs with strong expansion properties and the simple majority threshold. The convergence properties have also been studied for directed acyclic graphs, cf.~\cite{chistikov2020convergence}, and when a bias towards a superior opinion is present, cf.~\cite{lesfari2022biased}.

\textbf{Bounds.} There is a large body of research whose main goal is to find tight bounds on the minimum size of a TS in threshold models. Some prior work has investigated this for special classes of graphs, such as the $d$-dimensional torus~\cite{morris2009minimal} and random regular graphs~\cite{zehmakan2020opinion}. However, the main focus has been devoted to discovering sharp bounds in terms of various graph parameters such as the number of nodes~\cite{KSZ1,auletta2018reasoning}, girth~\cite{coja2014contagious}, maximum/minimum degree~\cite{FG}, expansion~\cite{zehmakan2020opinion}, vertex-transitivity~\cite{mossel2014majority}, and the minimum size of a feedback vertex set~\cite{adams2011dynamic}.

\textbf{Hardness.} The problem of determining the minimum size of a TS in the progressive threshold model has been investigated extensively, and it is known to be NP-hard even for some special choices of the threshold assignment and the input graph. Notably, it was proven in~\cite{C} that the problem cannot be approximated within the ratio of $\mathcal{O}(2^{\log^{1-\epsilon}n})$, for any fixed constant $\epsilon>0$, unless $NP\subseteq DTIME(n^{polylog(n)})$, even if we limit ourselves to simple majority threshold assignment and regular graphs.
Hardness results are also known for the non-progressive variant. For the simple majority threshold assignment, the problem cannot be approximated within a factor of $ \log \Delta \log \log \Delta$, unless $P=NP$, according to~\cite{mishra2002hardness}. For more hardness results also see~\cite{bredereck2021maximizing,nichterlein2013tractable,zehmakan2021majority}.

\textbf{Algorithms.} For certain classes of graphs and threshold assignments, the problem of finding the minimum size of a TS becomes tractable. For the progressive variant, it is proven that there is a polynomial time algorithm when the input graph is a tree, cf.~\cite{C,Centeno,Ben-Zwi}. This was generalized to block-cactus graphs in~\cite{chiang2013some}. The problem also is tractable when the feedback edge set number is small, cf.~\cite{nichterlein2013tractable}. Following up on an open problem from~\cite{Ben-Zwi}, recently it was shown in~\cite{BER} that in the non-progressive variant if $\tau(v)\in\{0,1,d(v)\}$ for every node $v$, then the problem is polynomial time solvable for trees. It is still open whether a polynomial time algorithm for the general threshold assignment on trees exists or not. (As mentioned in Section~\ref{our_contribution}, we provide a linear-time algorithm for TTS for any choice of the threshold assignment.) Furthermore, it is known that for the non-progressive model with the simple majority threshold, there exist a $(\log \Delta)$-approximation algorithm for general graphs~\cite{mishra2002hardness} and an exact linear-time solution for cycle and path graphs~\cite{DPRS}. Finally, the integer linear programming formulation of the problem has been studied by prior work, cf.~\cite{ackerman,SM,chen2023polyhedral}.

\section{Bounding Minimum Size of a TTS}\label{section Lower bound}
Let a \emph{disjoint} TTS be the same as a TTS except that a node cannot be targeted more than once (i.e., $\mathcal{S}_i\cap \mathcal{S}_j = \emptyset$ for every two distinct sets $\mathcal{S}_i$ and $\mathcal{S}_j$ in the definition of a TTS) and define $\overleftrightarrow{MDTT}(G,\tau)$ to be the minimum size of a disjoint TTS in the non-progressive model. In Theorem~\ref{disjoint_timed}, we prove that there are graphs and threshold assignments for which $\overleftrightarrow{MDTT}(G,\tau)$ is asymptotically smaller than $\overleftrightarrow{MT}(G,\tau)$. This indicates that from the two advantages that TTS has over TS (namely, (i) the
nodes can be targeted at different steps (ii) a node can be targeted more than once), advantage (i) solely suffices to make the manipulator substantially more powerful.

\begin{thm}\label{disjoint_timed}
There are arbitrarily large graphs $G$ such that $\overleftrightarrow{MT}(G,\tau)=\omega(\overleftrightarrow{MDTT}(G,\tau))$, where $\tau$ is strict majority.
\end{thm}
\textit{Proof Sketch.}
For an arbitrary integer $\kappa$, let set $L_i$, for $1\le i \le \kappa$, contains $i$ nodes. Then, add an edge between every node in $L_i$ and $L_{i+1}$, for $1\le i\le \kappa-1$. Finally, attach two leaves to the node in $L_1$ to construct a graph $G$ with $n=\Theta(\kappa^2)$ nodes. (See Figure~\ref{Tower Graph} for an example.)

	\begin{center}
		\begin{figure}
			\begin{center}
				\begin{tikzpicture}[scale=0.9]
					\draw(-0.5,0)--(0,1)--(0.5,0);
					\draw(-0.5,2)--(0,1)--(0.5,2);
					\draw[black,fill=blue](0,1)circle(0.1);
					\draw[black,fill=blue](-0.5,2)circle(0.1);
					\draw[black,fill=blue](0.5,2)circle(0.1);
					\draw(-0.5,0)--(-1,-1);
					\draw(-0.5,0)--(0,-1);
					\draw(-0.5,0)--(1,-1);
					\draw(0.5,0)--(-1,-1);
					\draw(0.5,0)--(0,-1);
					\draw(0.5,0)--(1,-1);
					\draw[black,fill=blue](-0.5,0)circle(0.1);
					\draw[black,fill=blue](0.5,0)circle(0.1);
					\draw[black,fill=blue](-1,-1)circle(0.1);
					\draw[black,fill=blue](0,-1)circle(0.1);
					\draw[black,fill=blue](1,-1)circle(0.1);
					\node[] at (-1,1){$L_1$};
					\draw[gray](0,1)ellipse(0.5 and 0.3);
					\node[] at (-1.5,0){$L_2$};
					\draw[gray](0,0)ellipse(1 and 0.3);
					\node[] at (-2,-1){$L_3$};
					\draw[gray](0,-1)ellipse(1.5 and 0.3);
					\node[] at (0,-2){$\vdots$};
					\node[] at (-3.6,-3){$L_{\kappa-1}$};
					\draw[gray](0,-3)ellipse(3 and 0.3);
					\node[] at (-4,-4){$L_{\kappa}$};
					\draw[gray](0,-4)ellipse(3.5 and 0.3);
					\draw(-2.5,-3)--(-3,-4);
					\draw(-2.5,-3)--(-2,-4);
					\draw(-2.5,-3)--(-1,-4);
					\draw(-2.5,-3)--(0,-4);
					\draw(-2.5,-3)--(3,-4);
					\draw(-1.5,-3)--(-3,-4);
					\draw(-1.5,-3)--(-2,-4);
					\draw(-1.5,-3)--(-1,-4);
					\draw(-1.5,-3)--(0,-4);
					\draw(-1.5,-3)--(3,-4);
					\draw(-0.5,-3)--(-3,-4);
					\draw(-0.5,-3)--(-2,-4);
					\draw(-0.5,-3)--(-1,-4);
					\draw(-0.5,-3)--(0,-4);
					\draw(-0.5,-3)--(3,-4);
					\draw(2.5,-3)--(-3,-4);
					\draw(2.5,-3)--(-2,-4);
					\draw(2.5,-3)--(-1,-4);
					\draw(2.5,-3)--(0,-4);
					\draw(2.5,-3)--(3,-4);
					\draw[black,fill=blue](-2.5,-3)circle(0.1);
					\draw[black,fill=blue](-1.5,-3)circle(0.1);
					\draw[black,fill=blue](-0.5,-3)circle(0.1);
					\node[] at (1,-3){$\cdots$};
					\draw[black,fill=blue](2.5,-3)circle(0.1);
					\draw[black,fill=blue](-3,-4)circle(0.1);
					\draw[black,fill=blue](-2,-4)circle(0.1);
					\draw[black,fill=blue](-1,-4)circle(0.1);
					\draw[black,fill=blue](0,-4)circle(0.1);
					\node[] at (1.5,-4){$\cdots$};
					\draw[black,fill=blue](3,-4)circle(0.1);
				\end{tikzpicture}
			\end{center}
			\caption{A graph where $\protect\overleftrightarrow{MT}(G,\tau)=\omega(\protect\overleftrightarrow{MDTT}(G,\tau))$.}
			\label{Tower Graph}
		\end{figure}
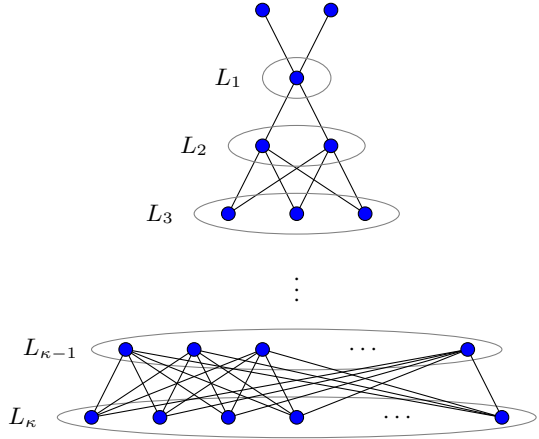
	\end{center}
 
Suppose that $\mathcal{S}_0=L_{\kappa}\cup L'_{\kappa-1}$ in which $L'_{\kappa-1}$ is any subset of $L_{\kappa-1}$ of cardinality $\left\lceil\tfrac{\kappa}{2}\right\rceil$. Also assume that $\mathcal{S}_{\kappa-2}$ consists of the node of $L_1$ and one of its leaf neighbors. For $1\le i\neq \kappa-2\le \kappa-1$ set $\mathcal{S}_i=\emptyset$. It is straightforward to check that $\mathcal{S}_0,\mathcal{S}_1,\ldots,\mathcal{S}_{\kappa-1}$ is a disjoint TTS of size $\kappa+\left\lceil\tfrac{\kappa}{2}\right\rceil+2$. This implies that $\overleftrightarrow{MDTT}(G,\tau)=\mathcal{O}(\sqrt{n})$.

Let $\mathcal{S}$ be a TS in the non-progressive model with strict majority. We claim that $|(L_{i-1}\cup L_{i+1})\cap \mathcal{S}|$ must be at least $i+1$, for $2\le i \le \kappa-1$. Thus, we have $|\mathcal{S}|=\Omega(n)$, which implies that $\overleftrightarrow{MT}(G,\tau)=\Omega(n)$. For the sake of contradiction, assume that $|(L_{i-1}\cup L_{i+1})\cap \mathcal{S}|< i+1$ for some $2\le i\le \kappa-1$. Consider the initial configuration where only the nodes in $\mathcal{S}$ are positive. We observe that $L_i$ becomes fully negative after one step. One step after that, $L_{i-1}$ becomes fully negative and so on until the only node in $L_1$ is negative. Once this happens, in all the following steps either the node in $L_1$ or its two leaf neighbors will be negative, regardless of the state of other nodes. This is in contradiction with $\mathcal{S}$ being a TS. (Please refer to Appendix~\ref{appendix-disjoint-TS}, for a full proof.) $\square$

\subsection{General Graphs}\label{general_graphs}
We first prove that $\overleftrightarrow{MTT}(G,\tau)\geq\frac{2n}{\Delta+1}$ for the strict majority threshold assignment if $G$ is bipartite in Theorem~\ref{main bound} (which is based on Lemma~\ref{Lem2}). Then, we provide Theorem~\ref{bipartite} which sets a connection between the value of $\overleftrightarrow{MTT}(G,\tau)$ in bipartite graphs and general graphs. Combining these two theorems gives us our desired bound for general graphs in Theorem~\ref{bound}.
The full proofs for these theorems are given in~\ref{appendix-general}.

\begin{lemma}\label{Lem2}
Let $\mathcal{S}_0,\mathcal{S}_1,\ldots,\mathcal{S}_k$ be a TTS of $(G,\tau)$ for a bipartite graph $G$ with partite sets $X$ and $Y$. Define $\mathcal{S}_e:=\mathcal{S}_0\cup \mathcal{S}_2\cup \mathcal{S}_4\cup \cdots$ and $\mathcal{S}_o:=\mathcal{S}_1\cup \mathcal{S}_3\cup \mathcal{S}_5\cup \cdots$. If for $D\subseteq V$ we have $D\cap \mathcal{S}_e\cap X=\emptyset$ and $D\cap \mathcal{S}_o\cap Y=\emptyset$ (or $D\cap \mathcal{S}_o\cap X=\emptyset$ and $D\cap \mathcal{S}_e\cap Y=\emptyset$), then $\left|E(G[D])\right|\leq\sum_{u\in D} \left(d(u)-\tau(u)\right)$.
\end{lemma}
\textit{Proof Sketch.}
Let us first prove Claim 1, which is the main ingredient of this proof.

\textbf{Claim 1.} \textit{If $D$ is nonempty, then there is a node $u\in D$ such that $d_D(u)\leq d(u)-\tau(u)$.}

Assume that $D\cap \mathcal{S}_e\cap X=\emptyset$ and $D\cap \mathcal{S}_o\cap Y=\emptyset$ (the proof is analogous for the other case). For the sake of contradiction, assume that $d_D(u)> d(u)-\tau(u)$ for all $u\in D$. One can show that for every odd $i$, $(D\cap Y)\cap \mathcal{A}_i=\emptyset$ and for every even $i$, $(D\cap X)\cap \mathcal{A}_i=\emptyset$. Since $D\neq\emptyset$, this contradicts the fact that $\mathcal{S}_0,\mathcal{S}_1,\ldots,\mathcal{S}_k$ is a TTS. This finishes the proof of Claim 1.

The statement of the lemma is trivial for $D=\emptyset$. For $|D|\geq1$ the proof is by induction on $|D|$. The base case of $|D|=1$ is straightforward. Assume that the inequality holds for $|D|<k$. Using Claim 1, there exists a node $v\in D$ such that $d_D(v)\leq d(v)-\tau(v)$. Applying the induction hypothesis for $D':=D\setminus \{v\}$ and some small calculations finish the proof. $\square$

\begin{thm}\label{main bound}
$\overleftrightarrow{MTT}(G,\tau)\geq\frac{2n}{\Delta+1}$ if $G$ is bipartite and $\tau$ is the strict majority.
\end{thm}
\textit{Proof Sketch.}
Let $X$ and $Y$ be the partite sets of $G$. Assume that $\mathcal{S}_0,\mathcal{S}_1,\ldots,\mathcal{S}_k$ is a TTS of $(G,\tau)$. Let $\mathcal{S}_o$ and $\mathcal{S}_e$ be as defined in Lemma~\ref{Lem2}. Furthermore, define $S:=\mathcal{S}_o\cup \mathcal{S}_e,~~S':=\mathcal{S}_o\setminus \mathcal{S}_e,~~S'':=\mathcal{S}_e\setminus \mathcal{S}_o,~~S''':=\mathcal{S}_o\cap \mathcal{S}_e$. Let $\mathcal{S}_X:=S\cap X,~~\mathcal{S}^{\prime}_X:=S'\cap X,~~S''_X:=S''\cap X,~~S'''_X:=S'''\cap X$ (and 
similarly, define $\mathcal{S}_Y$, $\mathcal{S}^{\prime}_Y$, $\mathcal{S}^{\prime\prime}_Y$, and $\mathcal{S}^{\prime\prime\prime}_Y$).
Let $F_1:=I\cup \mathcal{S}^{\prime}_X\cup S''_Y$ and $F_2:=I\cup S''_X\cup \mathcal{S}^{\prime}_Y$, where $I:=V\setminus S$. Both $F_1$ and $F_2$ clearly satisfy the conditions of Lemma~\ref{Lem2}. Combining the two inequalities obtained from applying Lemma~\ref{Lem2} and some calculations, we get:

\begin{equation}\label{general-eq}
\begin{split}
&\sum_{u\in I}\left(2\tau(u)-d(u)\right)\leq \sum_{u\in S}\left(d(u)-\tau(u)\right)+\sum_{u\in S'''} \tau(u).
\end{split}
\end{equation}

Using the fact that for the strict majority threshold assignment, $\tau(u)\geq\frac{d(u)+1}{2}$ and some further calculations, we can conclude that $\frac{2n}{\Delta+1} \leq \left|S'\cup S''\right| +2\left|S'''\right|\leq \overleftrightarrow{MTT}(G,\tau)$. $\square$

\begin{thm}\label{bipartite}
Let $G$ be a graph with node set $V(G):=\{v_1,v_2,\ldots,v_n\}$ and $\tau$ be a threshold assignment of its nodes. Construct the bipartite graph $H$ with partite sets $X:=\{x_1,x_2,\ldots,x_n\}$ and $Y:=\{y_1,y_2,\ldots,y_n\}$ whose edge set is $E(H):=\{x_iy_j|v_iv_j\in E(G)\}$. Consider threshold assignment $\tau'$ for $H$ with $\tau'(x_i)=\tau'(y_i)=\tau(v_i)$ for $1\leq i\leq n$. Then, 
$\overleftrightarrow{MTT}(H,\tau^{\prime})= 2\overleftrightarrow{MTT}(G,\tau)$.
\end{thm}

Combining Theorems~\ref{main bound} and~\ref{bipartite} gives us Theorem~\ref{bound}.
\begin{thm}\label{bound}
$\overleftrightarrow{MTT}(G, \tau)\geq\frac{2n}{\Delta(G)+1}$ if $\tau$ is the strict majority.
\end{thm}
\textbf{Tightness.} According to Example~\ref{example}, $\overleftrightarrow{MTT}(K_{1,n-1},\tau)=2$ when $\tau$ is the strict majority. This implies that the bound of $\overleftrightarrow{MTT}(K_{1,n-1},\tau)\geq\frac{2n}{\Delta(K_{1,n-1})+1}=\frac{2n}{(n-1)+1}=2$ is tight.

\subsection{Even Graphs}\label{even_graphs}
 
\begin{thm}\label{bound for general even}
$\overleftrightarrow{MTT}(G,\tau)\geq\frac{4n}{\Delta+2}$ if $G$ is even and $\tau$ is the strict majority, and this bound is tight.
\end{thm}
\textit{Proof Sketch.}
It suffices to prove the bound for bipartite graphs. Then, we can apply Theorem~\ref{bipartite} to extend to general graphs. Thus, let $G$ be bipartite.

Since $d(u)$ is an even number for any node $u$, we have $\tau(u)=\left\lceil\frac{d(u)+1}{2}\right\rceil=\frac{d(u)+2}{2}$ which implies $2\tau(u)-d(u)= 2$. Plugging this into Equation~\eqref{general-eq} in the proof of Theorem~\ref{main bound} yields $2|I|\leq \sum_{u\in S}\left(d(u)-\tau(u)\right)+\sum_{u\in S'''} \tau(u)$. Executing some calculations similar to the proof of Theorem~\ref{main bound}, we can show that $\frac{4n}{\Delta+2} \leq \left|S'\cup S''\right| +2\left|S'''\right|\leq \overleftrightarrow{MTT}(G,\tau)$. For a full proof of the theorem (including its tightness), see Appendix~\ref{appendix-even}. $\square$

\section{Algorithms to Find Minimum TTS}\label{algorithms}
We first prove that the \textsc{Timed Target Set Selection} problem is hard to approximate. In Subsections~\ref{Integer programming approach} and~\ref{hardness_greedy}, we provide an ILP formulation and propose a greedy algorithm for the problem, respectively. Then, we provide our experimental findings for real-world and synthetic graph data in Subsection~\ref{sec:experiments}. Finally, we present our linear time algorithm for trees in Subsection~\ref{Timed Target Set in Trees}.\\

\noindent\textsc{Timed Target Set Selection}
\\
\textbf{Input}: A graph $G$ and a threshold assignment $\tau$.
\\
\textbf{Output}: The minimum size of a TTS, i.e., $\overleftrightarrow{MTT}(G,\tau)$. \\

\subsection{Hardness Result}
\label{sec:hardness}

\begin{thm}
\label{hardness-thm}
The \textsc{Timed Target Set Selection } problem cannot be approximated within a ratio of $\mathcal{O}(2^{\log^{1-\epsilon}n})$ for any fixed $\epsilon>0$, unless $NP\subseteq DTIME(n^{polylog(n)})$.
\end{thm}
\textit{Proof Sketch.}	
For a given pair $(H,\tau')$, construct $(G,\tau)$ as follows. For each node $v\in V(H)$, add $\left\lceil\frac{d(v)}{2}\right\rceil$ copies of the complete graph $K_2$ and connect both nodes of each copy to $v$. Set $\tau(v)=\tau^{\prime}(v)$ if $v\in V(H)$ and $\tau(v)=1$ otherwise. We claim that $\overrightarrow{MT}(H,\tau')=\overleftrightarrow{MTT}(G,\tau)$.

Assume that there is a polynomial time algorithm for finding $\overleftrightarrow{MTT}(G,\tau)$ with the approximation ratio $C2^{\log^{1-\epsilon}|V(G)|}$ for some constants $C,\epsilon>0$. Since $|V(G)|=\mathcal{O}(|V(H)|^2)$, the above polynomial time reduction gives us a polynomial time algorithm with approximation ratio $C'2^{\log^{1-\epsilon'}|V(H)|}$, for some constants $\epsilon',C'>0$, for the problem of finding the minimum size of a TS in the progressive model. However, this is not possible according to the results from~\cite{C}, unless $NP\subseteq DTIME(n^{polylog(n)})$. (Please refer to Appendix~\ref{appendix-hardness} for a full proof.)
$\square$

\subsection{Integer Linear Programming Formulation}\label{Integer programming approach}
Here, we provide an Integer Linear Program (ILP) formulation for the \textsc{Timed Target Set Selection} problem. The binary variables $x_{vi}$ and $y_{vi}$ stand for the state of node $v$ in time step $i$. We have $x_{vi}=1$ if and only if $v\in \mathcal{S}_i$ and $y_{vi}=1$ if and only if $v\in Q_i$. So $x_{vi}+y_{vi}\geq 1$ if and only if $v$ is positive in time step $i$, i.e., $v\in\mathcal{A}_i$. Let $\mathbb{K}:=\{1,\cdots, k\}$ and $\mathbb{K}_0:=\mathbb{K}\cup\{0\}$. Furthermore, let us define the constraints
\begin{itemize}
    \item $C:=(d(v)+1-\tau(v))y_{vi}+\tau(v)-1 \ge \sum_{u\in N(v)} (x_{u(i-1)}+ y_{u(i-1)})$
    \item $C':=\sum_{u\in N(v)}(x_{u(i-1)}+y_{u(i-1)})- \tau(v)y_{vi}\geq 0$.
\end{itemize}
\begin{equation}\label{IP}
	\begin{array}{lll}
		\min & \sum_{i=0}^{k}\sum_{v\in V}x_{vi} \\
		{\rm s.t.} & C & \forall v\in V~~\forall i\in \mathbb{K} \\
		& C' & \forall v\in V~~\forall i\in \mathbb{K} \\
		& x_{vk}+y_{vk}=1 & \forall v\in V \\
            & x_{vi}+y_{vi}\leq 1 & \forall v\in V, i\in\mathbb{K}_0 \\
		& y_{v0}=0 & \forall v\in V \\
		& x_{vi}\in\{0,1\} & \forall v\in V~~\forall i\in \mathbb{K}_0\\
		& y_{vi}\in\{0,1\} & \forall v\in V~~\forall i\in \mathbb{K}_0
	\end{array}
\end{equation}

If the number of positive neighbors of a node $v$ in time step $i-1$ is greater than or equal to $\tau(v)$, then $v$ becomes positive in time step $i$. This is expressed as constraint $C$ in the ILP. Note that for $y_{vi}=0$, the constraint $C$ is equal to $\tau(v)>\sum_{u\in N(v)} (x_{u(i-1)}+ y_{u(i-1)})$. Furthermore, if the number of positive neighbors of a node $v$ in time step $i-1$ is strictly less than $\tau(v)$, then $v$ becomes negative in time step $i$ (unless we force it to be positive i.e., $x_{vi}=1$). This is expressed as the constraint $C'$. We observe that if $y_{vi}=1$, then the constraint $C'$ is equal to $\sum_{u\in N(v)}(x_{u(i-1)}+y_{u(i-1)})\ge \tau(v)$. The other constraints and the objective function are self-explanatory. 

The ILP~\eqref{IP} formulates the problem of finding the minimum size of a TTS which reaches the fully positive configuration in $k$ steps. To prove this, we need to show that a solution to the ILP corresponds to a TTS of the same size and vice versa. A formal proof of this is given in Appendix~\ref{appendix-ilp}, which builds on the observations from above on the connections between the constraints of the ILP and the updating rules.

It is known~\cite{poljak1986pre} that the non-progressive model stabilizes in $\mathcal{O}(n^2)$ steps. Putting this in parallel with the fact that a minimum TTS is of size at most $n$, we conclude that there is always a minimum TTS which reaches the fully positive configuration in $\mathcal{O}(n^3)$ steps. (Some details are left out.) Thus, the ILP~\eqref{IP} can be used to solve the 
 \textsc{Timed Target Set Selection} problem by ranging over different values of $k$.

\subsection{Greedy Algorithm}\label{hardness_greedy}
We provide a greedy algorithm which finds a TTS $\mathcal{S}_0$, $\mathcal{S}_1, \emptyset$ for a given pair $(G,\tau)$. The algorithm first sorts the nodes in ascending order of their degrees as $v_1,\cdots, v_n$ and set $\mathcal{S}_0=\emptyset, \mathcal{S}_1=\emptyset$. Then, nodes are processed one by one, and they are decided to be in $\mathcal{S}_0$, in $\mathcal{S}_1$, or in neither of them. The processed nodes which are in neither $\mathcal{S}_0$ nor $\mathcal{S}_1$ are called \textit{unselected} nodes. When processing a node $v_i$, a neighbor $u\in N(v_i)$ is said to be \textit{blocked} if the number of unselected nodes in $N(u)$ is equal to $d(u)-\tau(u)$ (i.e., if $v_i$ is set to be unselected, $u$ cannot become positive in the first step when only the nodes in $\mathcal{S}_0$ are positive). Let $\texttt{blocked}[v_i]$ be the set of blocked nodes in $N(v_i)$. If $|\texttt{blocked}[v_i]|=0$, we can safely set $v_i$ to be unselected. If $|\texttt{blocked}[v_i]|>1$, then we set $v_i$ to be in $\mathcal{S}_0$. If $|\texttt{blocked}[v_i]|=1$, we could either add $v_i$ to $\mathcal{S}_0$ or set $v_i$ to be unselected. For the latter, we add the only node in $\texttt{blocked}[v_i]$, say $w$, to $\mathcal{S}_1$. If $d(w)>d(v_i)$, we do the latter (since $w$ is more ``influential''), otherwise we execute the former. A precise description of the algorithm is given in Algorithm~\ref{modified greedy}. It is straightforward to prove that Algorithm~\ref{modified greedy} returns a TTS, using an inductive argument. For the sake of completeness, a formal proof is given in Appendix~\ref{appendix-greedy-correctness}.

\begin{algorithm}[t]
	\caption{Greedy Algorithm for TTS}
        \label{modified greedy}
	Sort the nodes of $G$ in ascending order of their degrees as the sequence $v_1,\ldots,v_n$.
	
	Set $\mathcal{S}_0=\emptyset$, $\mathcal{S}_1=\emptyset$, and ${\normalfont\texttt{unselected}}[v_i]=0$, ${\normalfont\texttt{blocked}}[v_i]=\emptyset$ for all $1\le i\le n$.

	\For{$i=1$ {\normalfont to} $n$}{
		\For{ $u\in N(v_i)$}{
			\uIf{${\normalfont\texttt{unselected}}[u]=d(u)-\tau(u)$}{add $u$ to ${\normalfont\texttt{blocked}}[v_i]$
			}
		}

		\uIf{$|{\normalfont\texttt{blocked}}[v_i]|=0$}{\For{$u\in N(v_i)$}{${\normalfont\texttt{unselected}}[u]+=1$}}
  
		\uIf{${|\normalfont\texttt{blocked}}[v_i]|>1$}{Add $v_i$ to $\mathcal{S}_0$}
		\uIf{${|\normalfont\texttt{blocked}}[v_i]|=1$}{\For{$w\in {\normalfont\texttt{blocked}}[v_i]$} {\uIf{$d(w)>d(v_i)$} {\For{$u\in N(v_i)$} {${\normalfont\texttt{unselected}}[u]+=1$} Add $w$ to $\mathcal{S}_1$, set $\tau(w)=0$}
				
				\Else{Add $v_i$ to $\mathcal{S}_0$}}}
	}
\end{algorithm}

Our algorithm is inspired by the well-known greedy algorithm for finding a TS (cf.~\cite{FG}) where there is no $\mathcal{S}_1$ and node $v_i$ is assigned to $\mathcal{S}_0$ if $|\texttt{blocked}[v_i]|\ge 1$ and set to be unselected otherwise. The final $\mathcal{S}_0$ is a TS which make all nodes positive in one step of the non-progressive model.

We did not provide any approximation guarantee for our algorithm, but according to Theorem~\ref{hardness-thm} we cannot hope for an approximation ratio better than $\mathcal{O}(2^{\log^{1-\epsilon}n})$ for any fixed $\epsilon>0$. Furthermore, our algorithm takes advantage of timing only for one extra step, but as we observe in the next subsection this would already allow us to produce solutions close to optimal. Devising algorithms which fully take advance of timing is left to future work.

\subsection{Experiments}
\label{sec:experiments}

Our experiments were carried out on an Intel Xeon E3 CPU, with 32 GB RAM, and a Linux OS and the code is in C++ and Python.

\textbf{Synthetic Networks.} We conducted experiments on BA (Barab\'{a}si–Albert) model~\cite{barabasi1999emergence} and ER (Erd\H{o}s–R\'{e}nyi) model~\cite{erdHos2013spectral}. The number of nodes were set to $n=10,15,\cdots, 40$ and the edge parameter, in both models, was set to obtain average degree $8$ (which is comparable to the average degree of the real-world networks of similar size, for example, see Karate club network~\cite{kunegis2013konect}). For each value of $n$, 10 instances of the random graph (BA and ER) were generated. Then, the optimal size of a TS and TTS and the outcome of the original greedy algorithm for TS and our algorithm for TTS were computed. The optimal solutions were found using the ILP formulations and a CBC solver~\cite{forrest2005cbc}. We considered the strict majority threshold assignment in all these experiments. The average outcome of each approach over all 10 graph instances are reported in Figure~\ref{opt:fig}. (The standard deviations were very small, namely smaller than $0.7$ in all cases.) Let us make the following two observations from Figure~\ref{opt:fig}:
\begin{itemize}
    \item \textbf{Observation 1.} Our proposed greedy algorithm performs very well and returns a solution very close to the optimal one in most cases.
    \item \textbf{Observation 2.} In most cases, the minimum size of a TTS (retuned by TTS-OPT) is strictly smaller than the minimum size of a TS (returned by TS-OPT). This confirms that the effect of timing is not limited to theoretically tailored graphs (such as the ones given in Example~\ref{example} and Theorem~\ref{disjoint_timed}).
\end{itemize}
\begin{figure}
\centerline{\includegraphics[height=3.2in]{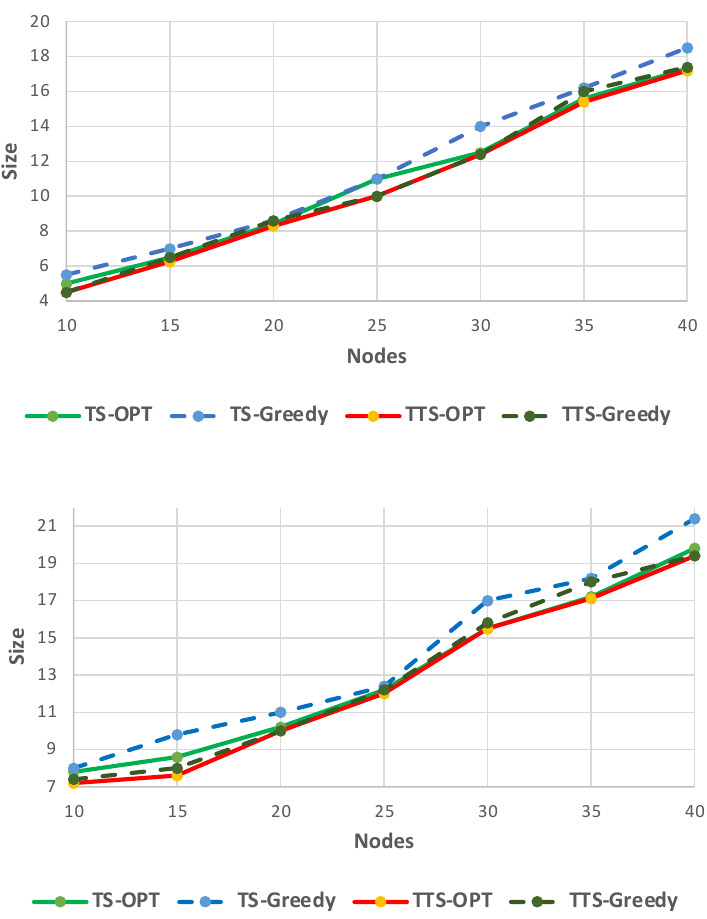}}
\caption{The average size of a TS/TTS retuned by ILP (TS-OPT and TTS-OPT) and greedy approach (TS-Greedy and TTS-Greedy) on BA graphs (top) and ER graphs (bottom). \label{opt:fig}}
\end{figure}

\textbf{Real-world Networks.}
We also have run our algorithm and the original greedy algorithm on different real-world networks, namely Twitter, Facebook, and Twitch Games from~\cite{snapnets} where we removed edge directions for Twitter. We again used the strict majority threshold assignment. (Note that computing the optimal solutions are not possible due to the large size of networks. However, based on diagrams in Figure~\ref{opt:fig}, the greedy algorithms seem to return solutions close to the optimal ones.) The outcomes, presented in Table~\ref{table}, demonstrate that allowing a manipulator who conducts a greedy approach (which is the most commonly proposed mechanism in the literature for various target set selection problems, cf.~\cite{kempe,FG}) to target nodes at their desired time gives them more than $13\%$ advantage for the graphs under experiment, indicating another time that timing matters, even on graph which emerge in the real world.

\begin{table}
\centering
\begin{tabular}{ |l|c|c|c|c|c| } 
\hline
Network & Nodes & TTS-Greedy & TS-Greedy & Imp. \\
\hline
Facebook & 4039 & 1727 & 1985 & 13\%   \\ 
Twitter & 81306 & 25022 & 28991 & 13.7\% \\ 
Twitch Games & 168114 & 44795 & 53726 & 16.6\%  \\ 
\hline
\end{tabular}
\caption{The size of a TS obtained by the greedy algorithm versus the size of a TTS from our greedy approach (Algorithm~\ref{modified greedy}), for the strict majority, and the improvement (Imp.) in the form of (TS-TTS)/TS.}
\label{table}
\end{table}


\subsection{Exact Linear-time Algorithm for Trees}\label{Timed Target Set in Trees}

Our goal in this section is to provide an exact algorithm for the \textsc{Timed Target Set Selection } problem for a given tree $T$ and threshold assignment $\tau$.

Let $\mathcal{L}(v)$ and $\bar{\mathcal{L}}(v)$ be the set of leaf and non-leaf neighbors of $v$, and define $l(v):=|\mathcal{L}(v)|$ and $\bar{l}(v):=|\bar{\mathcal{L}}(v)|$. Furthermore, we define $\mathcal{L}[v]:=\mathcal{L}(v)\cup\{v\}$ and $\bar{\mathcal{L}}[v]:=\bar{\mathcal{L}}(v)\cup\{v\}$. We partition the non-leaf nodes of $T$ into three sets: 
\begin{itemize}
\item $A:=\left\{v\in V(T):d(v)>1, \tau(v)>\bar{l}(v)\right\}$ 
\item $B:=\left\{v\in V(T):d(v)>1, \tau(v)<\bar{l}(v)\right\}$
\item $C:=\left\{v\in V(T):d(v)>1, \tau(v)=\bar{l}(v)\right\}$.
\end{itemize}
Furthermore, we partition $A$ into $A':=\{v\in A:\tau(v)<d(v)\}$ and $A'':=\{v\in A:\tau(v)=d(v)\}$.

Let us set a root $v$ for $T$. Then, we can partition $V(T)$ into $L_0\cup L_1\cup\cdots\cup L_d$, where the $i$-th level $L_i$ consists of nodes whose distance from $v$ is equal to $i$ and $d$ is the depth of the tree. Set $L_{-1}=\emptyset$.

Let us first consider the special case of $|A^{\prime\prime}\cup C|\le 1$. Set the only node in $A^{\prime\prime}\cup C$ (or an arbitrary non-leaf node if $A^{\prime\prime}\cup C=\emptyset$) as the root. Consider the sequence $\mathcal{S}_i=(L_{d-i}\cup L_{d-i-1})\cap A$ for $0\le i\le d$. We can prove that this is a TTS by induction on $d$, where the base case of $d=1$ corresponds to a star graph (similar to Example~\ref{example}). Since each node in $A$ appears exactly twice in this TTS, its size is equal to $2|A|$. On the other hand, for any TTS $\mathcal{S}_0,\mathcal{S}_1,\ldots, \mathcal{S}_k$ and any node $v\in A$, we have $\sum_{i=0}^k \left|\mathcal{S}_i\cap \mathcal{L}[v]\right|\geq2$. This is true because if $\sum_{i=0}^{k}\left|\mathcal{S}_i\cap \mathcal{L}[v]\right|\leq1$, then $v$ and its leaf neighbors cannot become positive simultaneously in any step and this contradicts the assumption that $\mathcal{S}_0,\mathcal{S}_1,\ldots, \mathcal{S}_k$ is a TTS. (This uses that $\tau(v)>0$ for any node $v$, which can be assumed as we will explain.) Since for any two distinct nodes $u$ and $v$, we have $\mathcal{L}[u]\cap \mathcal{L}[v]=\emptyset$, we can conclude that $\overleftrightarrow{MTT}(T,\tau)\geq 2|A|$. Furthermore, we argued that there is a TTS of size $2|A|$. Thus, $\overleftrightarrow{MTT}(T,\tau)= 2|A|$, which explains the lines 3-5 in Algorithm~\ref{Algorithm3}.

Now, let's consider the case of $|A^{\prime\prime}\cup C|\ge 2$. Again, we set a node $v\in A^{\prime\prime}\cup C$ as the root. Then, we iterate over the nodes respectively from $L_d$ to $L_1$. Let $T'$ be the subtree induced by the node $u$ which is being processed and its descendants. In the next two paragraphs, we explain that if $T'$ has a certain structure, then we know how to efficiently compute the minimum size of a TTS for $T'$ and how to connect it to the minimum size of a TTS in the original tree. This permits us to keep reducing the size of the tree while guaranteeing that the subtree induced by the next node that is processed will have one of our desired structures.

Suppose that $u$ (not equal to the root $v$) is a node in $A^{\prime\prime}$ and none of its descendants belongs to $A^{\prime\prime}\cup C$. Denote the parent of $u$ by $z$. Let $T'$ be the induced subtree on $u$ and its descendants and $T^{\prime\prime}$ be the induced subtree on $V(T)\setminus V(T')$. For $T'$, define $\tau'(w)=\tau(w)-1$ if $w=u$ and $\tau'(w)=\tau(w)$ otherwise. For $T^{\prime\prime}$, let $\tau^{\prime\prime}(w)=\tau(w)-1$ if $w=z$ and $\tau^{\prime\prime}(w)=\tau(w)$ otherwise. We can show that $\overleftrightarrow{MTT}(T, \tau)= \overleftrightarrow{MTT}(T',\tau^{\prime})+\overleftrightarrow{MTT}(T'',\tau^{\prime\prime})$. Using a proof similar to the case of $|A^{\prime\prime}\cup C|\le 1$ (which was handled before), we can prove that $\overleftrightarrow{MTT}(T',\tau')=2|V(T')\cap A|$. (There is a similar argument for $\tau(u)=0$.) This should explain the \textit{if} statement in line 11.

Suppose that $u$ is a non-root node in $C$ and none of its descendants belongs to $A^{\prime\prime}\cup C$. Furthermore, let $T'$ be the induced subtree on $u$ and its descendants and $T^{\prime\prime}$ be the induced subtree on $(V(T)\setminus V(T'))\cup \{u\}$. For $T'$, we define $\tau'(w)=\tau(w)-1$ if $w=u$ and $\tau'(w)=\tau(w)$ otherwise. For $T^{\prime\prime}$, we set $\tau^{\prime\prime}(w)=1$ if $w=u$ and $\tau^{\prime\prime}(w)=\tau(w)$ otherwise. We can prove that $\overleftrightarrow{MTT}(T,\tau)=\overleftrightarrow{MTT}(T',\tau')+\overleftrightarrow{MTT}(T^{\prime\prime},\tau^{\prime\prime})$. Again, using an argument similar to the case of $|A^{\prime\prime}\cup C|\le 1$, we can prove that $\overleftrightarrow{MTT}(T',\tau')=2|V(T')\cap A|$. This justifies the \textit{if} statement in line 14.

In Appendix~\ref{appendix-tree}, we provide a series of lemmas which explain Algorithm~\ref{Algorithm3} in a step by step and constructive fashion and prove its correctness.

	
	
	


\begin{algorithm}[t]
	\caption{Minimum Size of a TTS in Trees}
        \label{Algorithm3}
	Determine sets $A,A',A'',B,C$.
	
	Set $x=0$.
	
	\If{$|A''\cup C|\leq1$}{
         Set $x=2|A|$.
    }
	\Else{Choose some $v\in A''\cup C$ as the root. 
	
	\For{$i=d$ to $1$}
	{
		\For{$u$ in $L_i$}
		{
            Let $T'$ be the induced subtree on $u$ and its descendants, and $z$ be the parent of $u$.

			\If{$u\in A''$ or $\tau(u)=0$}{Set $x+=2|V(T')\cap A|$. Remove $V(T')$ and set $\tau(z)=\tau(z)-1$. Update sets $A,A',A'',B,C$.}
			\ElseIf{$u\in C$}
			{Set $x+=2|V(T')\cap A|$. Remove $V(T')\setminus\{u\}$ and set $\tau(u)=1$. Update sets $A,A',A'',B,C$.}
		}
	}
	Set $x+=2|V(T)\cap A|$.
	}
\end{algorithm}
\section{Future Work}


In~\cite{Ben-Zwi}, a polynomial time algorithm for finding $\overrightarrow{MT}(G,\tau)$ on trees was provided, and it was left open to determine whether finding $\overleftrightarrow{MT}(G,\tau)$ is tractable on trees. Recently, a polynomial time algorithm was given~\cite{BER} when $\tau(v)\in \{0,1,d(v)\}$ for any node $v$, but the problem remains open for general threshold assignment. We provided a linear-time algorithm to compute $\overleftrightarrow{MTT}(G,\tau)$ on trees for any threshold assignment. Can our techniques be leveraged to settle the problem for $\overleftrightarrow{MT}(G,\tau)$?

In practice, targeting some nodes, such as ``influencers'', might be more costly than others. A potential avenue for future research is to study the set-up where each node has a cost assigned to it and the manipulator aims to minimize the cost.

Finally, it would be interesting to devise new algorithms which take advantage of timing fully (unlike our greedy approach) and evaluate their performance on different real-world data.
\newpage
\bibliography{ref}

\begin{thebibliography}{10}

\bibitem{ackerman}
Eyal Ackerman, Oren Ben-Zwi, and Guy Wolfovitz, `Combinatorial model and bounds
  for target set selection', {\em Theoretical Computer Science}, {\bf
  411}(44-46),  4017--4022, (2010).

\bibitem{adams2011dynamic}
Sarah~Spence Adams, Denise~Sakai Troxell, and S~Luke Zinnen, `Dynamic
  monopolies and feedback vertex sets in hexagonal grids', {\em Computers \&
  Mathematics with Applications}, {\bf 62}(11),  4049--4057, (2011).

\bibitem{aggarwal2012influential}
Charu~C Aggarwal, Shuyang Lin, and Philip~S Yu, `On influential node discovery
  in dynamic social networks', in {\em Proceedings of the 2012 SIAM
  International Conference on Data Mining}, pp. 636--647. SIAM, (2012).

\bibitem{aizenman1988metastability}
Michael Aizenman and Joel~L Lebowitz, `Metastability effects in bootstrap
  percolation', {\em Journal of Physics A: Mathematical and General}, {\bf
  21}(19),  3801, (1988).

\bibitem{anagnostopoulosbiased}
Aris Anagnostopoulos, Luca Becchetti, Emilio Cruciani, Francesco Pasquale, and
  Sarai Rizzo, `Biased opinion dynamics: When the devil is in the details', in
  {\em IJCAI}, (2020).

\bibitem{auletta2018reasoning}
Vincenzo Auletta, Diodato Ferraioli, and Gianluigi Greco, `Reasoning about
  consensus when opinions diffuse through majority dynamics', in {\em {IJCAI}},
  (2018).

\bibitem{barabasi1999emergence}
Albert-L{\'a}szl{\'o} Barab{\'a}si and R{\'e}ka Albert, `Emergence of scaling
  in random networks', {\em science}, {\bf 286}(5439),  509--512, (1999).

\bibitem{Ben-Zwi}
Oren Ben-Zwi, Danny Hermelin, Daniel Lokshtanov, and Ilan Newman, `Treewidth
  governs the complexity of target set selection', {\em Discrete Optimization},
  {\bf 8}(1),  87--96, (2011).

\bibitem{bredereck2017manipulating}
Robert Bredereck and Edith Elkind, `Manipulating opinion diffusion in social
  networks', in {\em IJCAI}, (2017).

\bibitem{bredereck2021maximizing}
Robert Bredereck, Lilian Jacobs, and Leon Kellerhals, `Maximizing the spread of
  an opinion in few steps: opinion diffusion in non-binary networks', in {\em
  IJCAI}, pp. 1622--1628, (2021).

\bibitem{brill2016pairwise}
Markus Brill, Edith Elkind, Ulle Endriss, and Umberto Grandi, `Pairwise
  diffusion of preference rankings in social networks', in {\em IJCAI}, (2016).

\bibitem{castiglioni2020election}
Matteo Castiglioni, Diodato Ferraioli, and Nicola Gatti, `Election control in
  social networks via edge addition or removal', in {\em AAAI}, (2020).

\bibitem{Centeno}
Carmen~C Centeno, Mitre~C Dourado, Lucia~Draque Penso, Dieter Rautenbach, and
  Jayme~L Szwarcfiter, `Irreversible conversion of graphs', {\em Theoretical
  Computer Science}, {\bf 412}(29),  3693--3700, (2011).

\bibitem{chen2023polyhedral}
Cheng-Lung Chen, Eduardo~L Pasiliao, and Vladimir Boginski, `A polyhedral
  approach to least cost influence maximization in social networks', {\em
  Journal of Combinatorial Optimization}, {\bf 45}(1),  1--31, (2023).

\bibitem{C}
Ning Chen, `On the approximability of influence in social networks', {\em SIAM
  Journal on Discrete Mathematics}, {\bf 23}(3),  1400--1415, (2009).

\bibitem{chen2012time}
Wei Chen, Wei Lu, and Ning Zhang, `Time-critical influence maximization in
  social networks with time-delayed diffusion process', in {\em Proceedings of
  the AAAI Conference on Artificial Intelligence}, volume~26, pp. 591--598,
  (2012).

\bibitem{chiang2013some}
Chun-Ying Chiang, Liang-Hao Huang, Bo~Li~Jr, Jiaojiao Wu, and Hong-Gwa Yeh,
  `Some results on the target set selection problem', {\em Journal of
  Combinatorial Optimization}, {\bf 25}(4),  702--715, (2013).

\bibitem{chistikov2020convergence}
Dmitry Chistikov, Grzegorz Lisowski, Mike Paterson, and Paolo Turrini,
  `Convergence of opinion diffusion is pspace-complete', in {\em AAAI},
  volume~34, pp. 7103--7110, (2020).

\bibitem{coja2014contagious}
Amin Coja-Oghlan, Uriel Feige, Michael Krivelevich, and Daniel Reichman,
  `Contagious sets in expanders', in {\em SODA}, pp. 1953--1987, (2014).

\bibitem{coro2019exploiting}
F~Cor{\`o}, E~Cruciani, G~D'Angelo, and S~Ponziani, `Exploiting social
  influence to control elections based on scoring rules', in {\em IJCAI},
  volume~1, pp. 201--207, (2019).

\bibitem{DPRS}
Mitre~C Dourado, Lucia~Draque Penso, Dieter Rautenbach, and Jayme~L
  Szwarcfiter, `Reversible iterative graph processes', {\em Theoretical
  Computer Science}, {\bf 460},  16--25, (2012).

\bibitem{erdHos2013spectral}
L{\'a}szl{\'o} Erd{\H{o}}s, Antti Knowles, Horng-Tzer Yau, and Jun Yin,
  `Spectral statistics of erd{\H{o}}s--r{\'e}nyi graphs i: Local semicircle
  law', (2013).

\bibitem{faliszewskiopinion}
Piotr Faliszewski, Rica Gonen, Martin Kouteck{\`y}, and Nimrod Talmon, `Opinion
  diffusion and campaigning on society graphs', in {\em IJCAI}, (2018).

\bibitem{FG}
MohammadAmin Fazli, Mohammad Ghodsi, Jafar Habibi, Pooya Jalaly, Vahab
  Mirrokni, and Sina Sadeghian, `On non-progressive spread of influence through
  social networks', {\em Theoretical Computer Science}, {\bf 550},  36--50,
  (2014).

\bibitem{forrest2005cbc}
John Forrest and Robin Lougee-Heimer, `Cbc user guide', in {\em Emerging
  theory, methods, and applications},  257--277, INFORMS, (2005).

\bibitem{GOLES1980187}
Eric Goles and Jorge Olivos, `Periodic behaviour of generalized threshold
  functions', {\em Discrete mathematics}, {\bf 30}(2),  187--189, (1980).

\bibitem{kempe}
David Kempe, Jon Kleinberg, and {\'E}va Tardos, `Maximizing the spread of
  influence through a social network', in {\em SIGKDD}, pp. 137--146, (2003).

\bibitem{KSZ1}
Kaveh Khoshkhah, Hossein Soltani, and Manouchehr Zaker, `Dynamic monopolies in
  directed graphs: the spread of unilateral influence in social networks', {\em
  Discrete Applied Mathematics}, {\bf 171},  81--89, (2014).

\bibitem{kim2014ct}
Jinha Kim, Wonyeol Lee, and Hwanjo Yu, `Ct-ic: Continuously activated and
  time-restricted independent cascade model for viral marketing', {\em
  Knowledge-Based Systems}, {\bf 62},  57--68, (2014).

\bibitem{kunegis2013konect}
J{\'e}r{\^o}me Kunegis, `Konect: the koblenz network collection', in {\em
  Proceedings of the 22nd international conference on world wide web}, pp.
  1343--1350, (2013).

\bibitem{lesfari2022biased}
Hicham Lesfari, Fr{\'e}d{\'e}ric Giroire, and St{\'e}phane P{\'e}rennes,
  `Biased majority opinion dynamics: Exploiting graph $ k $-domination', in
  {\em IJCAI}, (2022).

\bibitem{snapnets}
Jure Leskovec and Andrej Krevl.
\newblock {SNAP Datasets}: {Stanford} large network dataset collection.
\newblock \url{http://snap.stanford.edu/data}, June 2014.

\bibitem{mishra2002hardness}
S~Mishra, Jaikumar Radhakrishnan, and Sivaramakrishnan Sivasubramanian, `On the
  hardness of approximating minimum monopoly problems', in {\em FSTTCS}, pp.
  277--288, (2002).

\bibitem{morris2009minimal}
Robert Morris, `Minimal percolating sets in bootstrap percolation.', {\em The
  Electronic Journal of Combinatorics [electronic only]}, {\bf 16}(1),
  Research--Paper, (2009).

\bibitem{mossel2014majority}
Elchanan Mossel, Joe Neeman, and Omer Tamuz, `Majority dynamics and aggregation
  of information in social networks', {\em AAMAS}, {\bf 28}(3),  408--429,
  (2014).

\bibitem{nichterlein2013tractable}
Andr{\'e} Nichterlein, Rolf Niedermeier, Johannes Uhlmann, and Mathias Weller,
  `On tractable cases of target set selection', {\em Social Network Analysis
  and Mining}, {\bf 3}(2),  233--256, (2013).

\bibitem{P}
David Peleg, `Size bounds for dynamic monopolies', {\em Discrete Applied
  Mathematics}, {\bf 86}(2-3),  263--273, (1998).

\bibitem{poljak1986pre}
Svatopluk Poljak and Daniel Turz{\'\i}k, `On pre-periods of discrete influence
  systems', {\em Discrete Applied Mathematics}, {\bf 13}(1),  33--39, (1986).

\bibitem{BER}
Dieter Rautenbach, Stefan Ehard, and Julien Baste, `Non-monotone target sets
  for threshold values restricted to $0 $, $1 $, and the vertex degree', {\em
  Discrete Mathematics \& Theoretical Computer Science}, {\bf 24}, (2022).

\bibitem{SM}
Hossein Soltani and Babak Moazzez, `A polyhedral study of dynamic monopolies',
  {\em Annals of Operations Research}, {\bf 279}(1),  71--87, (2019).

\bibitem{zehmakan2020opinion}
Ahad~N Zehmakan, `Opinion forming in erd{\H{o}}s--r{\'e}nyi random graph and
  expanders', {\em Discrete Applied Mathematics}, {\bf 277},  280--290, (2020).

\bibitem{zehmakan2021majority}
Ahad~N Zehmakan, `Majority opinion diffusion in social networks: An adversarial
  approach', in {\em AAAI}, volume~35, pp. 5611--5619, (2021).

\bibitem{zehmakan2023random}
Ahad~N Zehmakan, `Random majority opinion diffusion: Stabilization time,
  absorbing states, and influential nodes', {\em arXiv preprint
  arXiv:2302.06760}, (2023).

\bibitem{zhang2013maximizing}
Huiyuan Zhang, Thang~N Dinh, and My~T Thai, `Maximizing the spread of positive
  influence in online social networks', in {\em 2013 IEEE 33rd International
  Conference on Distributed Computing Systems}, pp. 317--326. IEEE, (2013).

\end{thebibliography}

\appendix

\newpage 
\section{Proof of Theorem~\ref{disjoint_timed}}
\label{appendix-disjoint-TS}

For an arbitrary integer $\kappa$, let set $L_i$, for $1\le i \le \kappa$, contain $i$ nodes. Then, add an edge between every node in $L_i$ and every node in $L_{i+1}$, for $1\le i\le \kappa-1$. Finally, attach two leaves to the node in $L_1$ to complete the construction of graph $G$. Note that the number of nodes in $G$ is equal to
\[
n=2+\sum_{i=1}^{\kappa}i=2+\frac{\kappa(\kappa+1)}{2}=\Theta(\kappa^2).
\]

Let us first provide a disjoint TTS in the non-progressive model with the strict majority threshold assignment. Suppose that $\mathcal{S}_0=L_{\kappa}\cup L'_{\kappa-1}$ where $L'_{\kappa-1}$ is any subset of $L_{\kappa-1}$ of cardinality $\left\lceil\tfrac{\kappa}{2}\right\rceil$. Also assume that $\mathcal{S}_{\kappa-2}$ consists of the node of $L_1$ and one of its leaf neighbors. For $1\le i\neq \kappa-2\le \kappa-1$ set $\mathcal{S}_i=\emptyset$. It is straightforward to check that $Q_1=L_{\kappa}\cup L_{\kappa-1}$, $Q_2=L_{\kappa}\cup L_{\kappa-1}\cup L_{\kappa-2}$, and finally $Q_{\kappa-1}=V$. Therefore, the sequence $\mathcal{S}_0,\mathcal{S}_1,\ldots,\mathcal{S}_{\kappa-1}$ is a disjoint TTS of size $\kappa+\left\lceil\tfrac{\kappa}{2}\right\rceil+2$. This implies that $\overleftrightarrow{MDTT}(G,\tau)=\mathcal{O}(\sqrt{n})$ since $\kappa=\Theta(\sqrt{n})$.

Let $\mathcal{S}$ be a minimum size TS in the non-progressive model and the threshold assignment $\tau$ be the strict majority. We claim that for each $2\leq i\leq \kappa-1$ we have $|(L_{i-1}\cup L_{i+1})\cap\mathcal{S}|\geq i+1$. For the sake of contradiction, assume that $|(L_{i-1}\cup L_{i+1})\cap\mathcal{S}| < i+1$ for some $2\le i\le \kappa-1$. Since the threshold of nodes belonging to $L_i$ is equal to $i+1$, in time step $1$ all nodes in $L_i$ become negative. The threshold of nodes belonging to $L_{i-1}$ is $i$ and they have $i-2$ neighbors outside $L_i$ (i.e.\,in $L_{i-2}$). Since $L_i$ is completely negative in time step $1$, so in time step $2$ all nodes in $L_{i-1}$ become negative. By repeating this argument, we conclude that in time step $i$, the only node in $L_1$, let's call it $v$, becomes negative. Once $v$ is negative, in all the upcoming steps $v$ or its two leaf neighbors will be negative. This is in contradiction with $\mathcal{S}$ being a TS. Therefore, we conclude that

\begin{eqnarray*}
		\overleftrightarrow{MT}(G,\tau) & =  & |\mathcal{S}| \\
		& \geq & \sum_{i=1}^{\kappa}|L_{i}\cap\mathcal{S}| \\
		& \geq & \frac{1}{2}\left(\sum_{i=2}^{\kappa-1}|L_{i-1}\cap\mathcal{S}| + \sum_{i=2}^{\kappa-1}|L_{i+1}\cap\mathcal{S}|\right) \\
		& = & \frac{1}{2} \sum_{i=2}^{\kappa-1}\left(|L_{i-1}\cap\mathcal{S}|+|L_{i+1}\cap\mathcal{S}|\right)\\
		& = & \frac{1}{2} \sum_{i=2}^{\kappa-1}|(L_{i-1}\cup L_{i+1})\cap\mathcal{S}| \\
		& \geq & \frac{1}{2} \sum_{i=2}^{\kappa-1} (i+1) \\
		& = & \frac{1}{2} \times \frac{(\kappa+3)(\kappa-2)}{2} \\
		& = & \Omega(\kappa^2) \\
  & = & \Omega(n).
\end{eqnarray*}

We proved that $\overleftrightarrow{MDTT}(G,\tau)=\mathcal{O}(\sqrt{n})$ and $\overleftrightarrow{MT}(G,\tau)=\Omega(n)$. Thus, we conclude that $\overleftrightarrow{MT}(G,\tau)=\omega(\overleftrightarrow{MDTT}(G,\tau))$.

\section{Full Proofs of Section~\ref{general_graphs}}\label{appendix-general}

\subsection{Proof of Lemma~\ref{Lem2}}

Let us first prove Claim 1, which is the main ingredient of this proof.

\textbf{Claim 1.} \textit{If $D$ is nonempty, then there is a node $u\in D$ such that $d_D(u)\leq d(u)-\tau(u)$.}

We assume that $D\cap \mathcal{S}_e\cap X=\emptyset$ and $D\cap \mathcal{S}_o\cap Y=\emptyset$. The proof is analogous for the case of $D\cap \mathcal{S}_o\cap X=\emptyset$ and $D\cap \mathcal{S}_e\cap Y=\emptyset$.

For the sake of contradiction, assume that for all $u\in D$ we have $d_D(u)> d(u)-\tau(u)$. We claim that any node $y\in D\cap Y$ does not become positive in step one, i.e., $y\notin \mathcal{S}_1\cup Q_1$. To prove the claim, first note that since $D\cap \mathcal{S}_o\cap Y=\emptyset$, we have $y\notin \mathcal{S}_o$. So the definition of $\mathcal{S}_o$ implies that $y\notin \mathcal{S}_1$. Furthermore, $N(y)\cap \mathcal{S}_0\subseteq X\cap \mathcal{S}_e$. So the assumption $D\cap \mathcal{S}_e\cap X=\emptyset$ implies $N(y)\cap \mathcal{S}_0\subseteq N(y)\setminus D$. So, we have $\left|N(y)\cap \mathcal{S}_0\right|  \leq \left|N(y)\setminus D\right| =  \left|N(y)\right| - \left|N(y)\cap D\right|
=  d(y)-d_D(y) < \tau(y)$. Then, according to the rule $y\notin Q_1$. This proves the claim and so we have $(D\cap Y)\cap (\mathcal{S}_1\cup Q_1)=\emptyset$.

Now, we claim that for all $x\in D\cap X$ we have $x\notin \mathcal{S}_2\cup Q_2$. To prove the claim first note that since $D\cap \mathcal{S}_e\cap X=\emptyset$, we have $x\notin \mathcal{S}_e$ and so $x\notin \mathcal{S}_2$. Furthermore, $N(x)\cap (\mathcal{S}_1\cup Q_1)\subseteq Y\cap (\mathcal{S}_1\cup Q_1)$. Since $(D\cap Y)\cap (\mathcal{S}_1\cup Q_1)=\emptyset$, we conclude that $N(x)\cap (\mathcal{S}_1\cup Q_1)\subseteq N(x)\setminus D$. So, we get $\left|N(x)\cap (\mathcal{S}_1\cup Q_1)\right|  \leq \left|N(x)\setminus D\right| = d(x)-d_D(x) <\tau(x)$. Then, according to the rule $x\notin Q_2$. This proves the claim and so we have $(D\cap X)\cap (\mathcal{S}_2\cup Q_2)=\emptyset$. This means that the nodes of $D\cap X$ do not become positive in time step $2$. 

By repeating this argument, we can conclude that for every odd $i$, $(D\cap Y)\cap (\mathcal{S}_i\cup Q_i)=\emptyset$ and for every even $i$, $(D\cap X)\cap (\mathcal{S}_i\cup Q_i)=\emptyset$. Since $D\neq\emptyset$, so $D\cap X\neq\emptyset$ or $D\cap Y\neq\emptyset$. W.l.o.g., suppose that $D\cap X\neq\emptyset$ and there exists a node $x\in D\cap X$. Moreover, the assumption that $\mathcal{S}_0,\mathcal{S}_1,\ldots,\mathcal{S}_k$ is a TTS of $G$ implies $Q_i=V$ for all time steps $i\geq k$. So $x\in D\cap X\subseteq V=Q_i$ for all $i\geq k$. Thus, $x\in (D\cap X)\cap (\mathcal{S}_i\cup Q_i)$ for all $i\geq k$. This contradicts the fact that $(D\cap X)\cap (\mathcal{S}_i\cup Q_i)=\emptyset$ for all even $i$. This finishes the proof of Claim 1.

Now, let us prove the statement of the lemma using Claim 1. The case of $D=\emptyset$ is trivial. For $|D|\geq1$ the proof is by induction on $|D|$. For $|D|=1$ the inequality is true since the left-hand side is zero and for all nodes the right-hand side is non-negative using the assumption $\tau(u)\leq d(u)$. 

Assume that the inequality holds for $|D|<k$, then we prove that it also holds for $|D|=k$. Using Claim 1, there exists a node $v\in D$ such that $d_D(v)\leq d(v)-\tau(v)$. The node set $D':=D\setminus \{v\}$ has the conditions of the lemma and its cardinality is $k-1$. So the induction hypothesis yields 
$|E(G[D'])|\leq\sum_{u\in D'} \left(d(u)-\tau(u)\right)$. So we conclude that
\begin{eqnarray*}
		|E(G[D])| & = & d_D(v)+|E(G[D'])|\\
		& \leq & d(v)-\tau(v)+\sum_{u\in D'} \left(d(u)-\tau(u)\right)\\
		& = & \sum_{u\in D} \left(d(u)-\tau(u)\right).
\end{eqnarray*}

\subsection{Proof of Theorem~\ref{main bound}}

Suppose that $X$ and $Y$ are the partite sets of $G$. Let $\mathcal{S}_0,\mathcal{S}_1,\ldots,\mathcal{S}_k$ be a TTS of $G$. Consider $\mathcal{S}_o$ and $\mathcal{S}_e$ as defined in Lemma~\ref{Lem2}. Define $S$, $S'$, $S''$ and $S'''$ as follows $S:= \mathcal{S}_e\cup \mathcal{S}_o$, $S':=\mathcal{S}_o\setminus \mathcal{S}_e,~~S'':=\mathcal{S}_e\setminus \mathcal{S}_o,~~S''':=\mathcal{S}_o\cap \mathcal{S}_e$. Furthermore, for partite sets $X$ and $Y$ we use the following notations for brevity: $\mathcal{S}_X:=S\cap X,~~\mathcal{S}^{\prime}_X:=S'\cap X,~~S''_X:=S''\cap X,~~S'''_X:=S'''\cap X$ and 
$\mathcal{S}_Y:=S\cap Y,~~\mathcal{S}^{\prime}_Y:=S'\cap Y,~~S''_Y:=S''\cap Y,~~S'''_Y:=S'''\cap Y$.
Let $I:=V\setminus S$. Now, consider $F_1:=I\cup \mathcal{S}^{\prime}_X\cup S''_Y$ and $F_2:=I\cup S''_X\cup \mathcal{S}^{\prime}_Y$ which both clearly satisfy the conditions of Lemma~\ref{Lem2}. So by applying Lemma~\ref{Lem2} for $F_1$, we have 
$$\left|E(G[I\cup \mathcal{S}^{\prime}_X\cup S''_Y])\right|\leq \sum_{u\in I\cup \mathcal{S}^{\prime}_X\cup S''_Y}\left(d(u)-\tau(u)\right)$$
and for $F_2$ we have 
$$\left|E(G[I\cup S''_X\cup \mathcal{S}^{\prime}_Y])\right|\leq \sum_{u\in I\cup S''_X\cup \mathcal{S}^{\prime}_Y}\left(d(u)-\tau(u)\right).$$
Let $A$ and $B$ be two disjoint subsets of $V$. Denote by $e_A$ and $e_{AB}$ the number of edges of $G[A]$ and the number of edges between $A$ and $B$, respectively. Incorporating these notations into the above two inequalities we obtain

\begin{multline*}
e_I+e_{I\mathcal{S}^{\prime}_X}+e_{IS''_Y}+e_{\mathcal{S}^{\prime}_XS''_Y}\leq \sum_{u\in I\cup \mathcal{S}^{\prime}_X\cup S''_Y}\left(d(u)-\tau(u)\right)
\end{multline*}
and 
\begin{multline*}
e_I+e_{IS''_X}+e_{I\mathcal{S}^{\prime}_Y}+e_{S''_X\mathcal{S}^{\prime}_Y}\leq \sum_{u\in I\cup S''_X\cup \mathcal{S}^{\prime}_Y}\left(d(u)-\tau(u)\right).
\end{multline*}

Combining these two inequalities and using the fact that $V=I\cup S'\cup S''\cup S'''$ is a partition for $V$ yield
\begin{multline*}
2e_I+(e_{I\mathcal{S}_X}-e_{IS'''_X})+(e_{I\mathcal{S}_Y}-e_{IS'''_Y})+e_{S'S''} \leq \\ \sum_{u\in I} \left(d(u)-\tau(u)\right)+\sum_{u\in V}\left(d(u)-\tau(u)\right)- \\ \sum_{u\in S'''} \left(d(u)-\tau(u)\right).
\end{multline*}

Clearly, we have $2e_I+e_{I\mathcal{S}_X}+e_{I\mathcal{S}_Y}=\sum_{u\in I}d(u)$. So, we conclude that 
\begin{multline*}
    \sum_{u\in I}d(u)-e_{IS'''}+e_{S'S''}\leq\sum_{u\in I} \left(d(u)-\tau(u)\right)+ \\ \sum_{u\in V}\left(d(u)-\tau(u)\right)-\sum_{u\in S'''} \left(d(u)-\tau(u)\right).
\end{multline*}
Since $e_{S'S''}$ is non-negative and $e_{IS'''}\leq \sum_{u\in S'''} d(u)$, we have
\begin{align*}
\sum_{u\in I}d(u)\leq\sum_{u\in I} \left(d(u)-\tau(u)\right)+ \sum_{u\in V}\left(d(u)-\tau(u)\right)\\  -\sum_{u\in S'''} \left(d(u)-\tau(u)\right)+\sum_{u\in S'''} d(u).
\end{align*}
Since $V=I\cup S$, we conclude that 
\begin{multline*}
\sum_{u\in I}d(u)\leq 2\sum_{u\in I}\left(d(u)-\tau(u)\right) \\ +\sum_{u\in S}\left(d(u)-\tau(u)\right)+ \sum_{u\in S'''} \tau(u)
\end{multline*}
and so 
\begin{equation}\label{general}
\begin{split}
\sum_{u\in I}\left(2\tau(u)-d(u)\right)\leq \sum_{u\in S}\left(d(u)-\tau(u)\right)+\sum_{u\in S'''} \tau(u).
\end{split}
\end{equation}
	
For the strict majority threshold assignment, we have $\tau(u)\geq\frac{d(u)+1}{2}$ which implies that $2\tau(u)-d(u)\geq 1$. Plugging this in Equation~\eqref{general} gives 
\begin{align*}
|I| \leq &\sum_{u\in I}\left(2\tau(u)-d(u)\right) \leq  \sum_{u\in S}\left(d(u)-\tau(u)\right) +\sum_{u\in S'''} \tau(u).
\end{align*}

Since $S=S'\cup S''\cup S'''$, we conclude that 
$$|I|\leq \sum_{u\in S'\cup S''}\left(d(u)-\tau(u)\right)+\sum_{u\in S'''} d(u).$$
Thus, we get
$$|I|\leq \sum_{u\in S'\cup S''}\left(\frac{d(u)-1}{2}\right)+\sum_{u\in S'''} d(u).$$
Therefore, we have

\begin{multline*}
 |I|\leq \left|S'\cup S''\right|\frac{\Delta-1}{2}+\left|S'''\right|\Delta \Rightarrow \\
 n-\left|S'\cup S''\right|-\left|S'''\right|\leq \left|S'\cup S''\right|\frac{\Delta-1}{2}+ \left|S'''\right|\Delta.
\end{multline*}
This implies that
\begin{multline*}
 n \leq \left|S'\cup S''\right|\frac{\Delta+1}{2}+\left|S'''\right|\left(\Delta+1\right) \Rightarrow \\
 \frac{2n}{\Delta+1} \leq \left|S'\cup S''\right| +2\left|S'''\right|\leq \overleftrightarrow{MTT}(G,\tau).
\end{multline*}

In the last inequality, we used the fact that every node in $S'\cup S''$ is enumerated at least once and every node in $S'''$ is enumerated at least twice (in an odd and also in an even time step) in the size of the TTS. 

\subsection{Proof of Theorem~\ref{bipartite}}

First, we show that $\overleftrightarrow{MTT}(H,\tau^{\prime})\leq 2\overleftrightarrow{MTT}(G,\tau)$. Suppose that $\mathcal{S}_0,\mathcal{S}_1,\ldots,\mathcal{S}_k$ is a minimum size TTS for $(G,\tau)$ and $Q_0,Q_1,\ldots,Q_k$ is the corresponding sequence of positive nodes. For every $0\leq j\leq k$ set $\mathcal{S}^{\prime}_j:=\{x_i|v_i\in \mathcal{S}_j\}\cup \{y_i|v_i\in \mathcal{S}_j\}$ and $Q^{\prime}_j:=\{x_i|v_i\in Q_j\}\cup \{y_i|v_i\in Q_j\}$. We prove by induction that the sequence $\mathcal{S}^{\prime}_0,\mathcal{S}^{\prime}_1,\ldots,\mathcal{S}^{\prime}_k$ is a TTS for $(H,\tau')$ with corresponding positive set sequence $Q'_0,Q'_1,\ldots,Q'_k$. The base of induction is trivial since $Q'_0=\emptyset$. As the induction hypothesis, assume that the statement holds for some $j\ge 1$. We show that $Q'_{j+1}=\{u\in V(H):|N(u)\cap(\mathcal{S}'_{j}\cup Q'_{j})|\geq\tau'(u)\}$. Note that $x_i\in Q'_{j+1}$ and $y_i\in Q'_{j+1}$ if and only if $v_i\in Q_{j+1}$ and since $\mathcal{S}_0,\mathcal{S}_1,\ldots,\mathcal{S}_k$ is a TTS of $G$ this is equivalent to $|N(v_i)\cap(\mathcal{S}_j\cup Q_j)|\geq \tau(v_i)$. By the construction of $(H,\tau')$ from $(G,\tau)$ this is equivalent to $|N(x_i)\cap(\mathcal{S}'_j\cup Q'_j)|\geq \tau(x_i)$ and $|N(y_i)\cap(\mathcal{S}'_j\cup Q'_j)|\geq \tau(y_i)$. Thus, we conclude that $\mathcal{S}^{\prime}_0,\mathcal{S}^{\prime}_1,\ldots,\mathcal{S}^{\prime}_k$ is a TTS of $(H,\tau')$ with size $2\overleftrightarrow{MTT}(G,\tau)$. This implies that $\overleftrightarrow{MTT}(H,\tau^{\prime})\leq 2\overleftrightarrow{MTT}(G,\tau)$.

Now, we show that $\overleftrightarrow{MTT}(H,\tau^{\prime})\geq 2\overleftrightarrow{MTT}(G,\tau)$. Let $\mathcal{S}^{\prime}_0,\mathcal{S}^{\prime}_1,\ldots,\mathcal{S}^{\prime}_k$ be a minimum size TTS of $(H,\tau')$. Suppose that $k$ is even because otherwise we can consider TTS $\mathcal{S}^{\prime}_0,\mathcal{S}^{\prime}_1,\ldots,\mathcal{S}^{\prime}_k,\mathcal{S}^{\prime}_{k+1}$ with $\mathcal{S}^{\prime}_{k+1}=\emptyset$. Let $\mathcal{S}^X_j=\mathcal{S}^{\prime}_j\cap X$ and $\mathcal{S}^Y_j=\mathcal{S}^{\prime}_j\cap Y$ for $0\leq j\leq k$. Consider the partitioning of $\mathcal{S}^{\prime}_0,\mathcal{S}^{\prime}_1,\ldots,\mathcal{S}^{\prime}_k$ into the sequences $\mathcal{S}^X_0,\mathcal{S}^Y_1,\mathcal{S}^X_2,\mathcal{S}^Y_3,\ldots,\mathcal{S}^Y_{k-1},\mathcal{S}^X_{k}$ and $\mathcal{S}^Y_0,\mathcal{S}^X_1,\mathcal{S}^Y_2,\mathcal{S}^X_3,\ldots,\mathcal{S}^X_{k-1},\mathcal{S}^Y_{k}$. W.l.o.g. suppose that

\begin{align*}
|\mathcal{S}^X_0|+|\mathcal{S}^Y_1|+|\mathcal{S}^X_2|+|\mathcal{S}^Y_3|+\ldots+|\mathcal{S}^Y_{k-1}|+|\mathcal{S}^X_{k}|\leq \\|\mathcal{S}^Y_0|+|\mathcal{S}^X_1|+|\mathcal{S}^Y_2|+|\mathcal{S}^X_3|+\ldots+|\mathcal{S}^X_{k-1}|+|\mathcal{S}^Y_{k}|.
\end{align*}
Since in time step $k$ all nodes of $H$ are made positive by the TTS, $X$ will be fully positive by $\mathcal{S}^X_0,\mathcal{S}^Y_1,\mathcal{S}^X_2,\mathcal{S}^Y_3,\ldots,\mathcal{S}^Y_{k-1},\mathcal{S}^X_{k}$ in time step $k$. For $0\leq j\leq k$, if $j$ is even, set $\mathcal{S}_j=\{v_i|x_i\in \mathcal{S}^X_j\}$ and if $j$ is odd, set $\mathcal{S}_j=\{v_i|y_i\in \mathcal{S}^Y_j\}$. It is straightforward to see that $\mathcal{S}_0,\mathcal{S}_1,\ldots,\mathcal{S}_k$ is a TTS of $(G,\tau)$. So, we conclude that $\overleftrightarrow{MTT}(G,\tau)$ is at most
\begin{eqnarray*}
&&|\mathcal{S}_0|+|\mathcal{S}_1|+\cdots+|\mathcal{S}_k|= \\
		&&|\mathcal{S}^X_0|+|\mathcal{S}^Y_1|+\cdots+|\mathcal{S}^Y_{k-1}|+|\mathcal{S}^X_{k}| \leq\\
		&&\frac{1}{2}\Big(|\mathcal{S}^X_0|+|\mathcal{S}^Y_1|+\cdots+|\mathcal{S}^Y_{k-1}|+|\mathcal{S}^X_{k}| \\
		& &+|\mathcal{S}^Y_0|+|\mathcal{S}^X_1|+\cdots+|\mathcal{S}^X_{k-1}|+|\mathcal{S}^Y_{k}| \Big) =\\
		&&\frac{1}{2}\Big(|\mathcal{S}^{\prime}_0|+|\mathcal{S}^{\prime}_1|+\cdots+|\mathcal{S}^{\prime}_k|\Big) =\\
		&&\frac{1}{2}\overleftrightarrow{MTT}(H,\tau^{\prime}).
\end{eqnarray*}

\subsection{Proof of Theorem~\ref{bound}}

Construct bipartite graph $H$ with threshold $\tau'$ from $(G,\tau)$ as Theorem~\ref{bipartite}. According to the construction of $H$, it is straightforward to observe that the threshold $\tau'$ is strict majority threshold with respect to $H$. Now, by applying Theorem~\ref{main bound}, we get $\overleftrightarrow{MTT}(H,\tau^{\prime})\geq\frac{2|V(H)|}{\Delta(H)+1}$. Theorem~\ref{bipartite} implies that $\overleftrightarrow{MTT}(G,\tau)=\overleftrightarrow{MTT}(H,\tau^{\prime})/2$. Since $|V(H)|=2|V(G)|=2n$ and $\Delta(H)=\Delta(G)$, we conclude that
$\overleftrightarrow{MTT}(G,\tau)\geq\frac{2n}{\Delta(G)+1}$.

\section{Proof of Theorem~\ref{bound for general even}}
\label{appendix-even}

It suffices to prove the bound for bipartite graphs. Then, we can apply Theorem~\ref{bipartite} to extend it to general graphs, similar to the proof of Theorem~\ref{bound}. Thus, assume that $G$ is bipartite, which allows us to use the proof of Theorem~\ref{main bound} and adjust it to get a stronger bound when $G$ is even.

Since $d(u)$ is an even number for all $u\in V(G)$, we have $\tau(u)=\left\lceil\frac{d(u)+1}{2}\right\rceil=\frac{d(u)+2}{2}$ and therefore $2\tau(u)-d(u)= 2$. By plugging this into Equation~\eqref{general} in the proof of Theorem~\ref{main bound} we get:
\begin{equation*}
		2|I|\leq \sum_{u\in S}\left(d(u)-\tau(u)\right)+\sum_{u\in S'''} \tau(u).
	\end{equation*}
	Since $S=S'\cup S''\cup S'''$, we conclude that 
	\begin{eqnarray*}
		2|I| & \leq & \sum_{u\in S'\cup S''}\left(d(u)-\tau(u)\right)+\sum_{u\in S'''} d(u) \\  
		& \leq & \sum_{u\in S'\cup S''}\left(\frac{d(u)-2}{2}\right)+\sum_{u\in S'''} d(u) \\
		& \leq & \left|S'\cup S''\right|\frac{\Delta(G)-2}{2}+\left|S'''\right|\Delta(G).
	\end{eqnarray*}
Since $V(G)$ consists of a disjoint union of $I$, $S'$, $S''$ and $S'''$, we have
\begin{align*}
& 2n-2\left|S'\cup S''\right|-2\left|S'''\right|\leq \\ & \left|S'\cup S''\right|\frac{\Delta(G)-2}{2}+\left|S'''\right|\Delta(G).
\end{align*}
This implies that
\begin{align*}
    2n \leq \left|S'\cup S''\right|\frac{\Delta(G)+2}{2} +\left|S'''\right|\left(\Delta(G)+2\right).
\end{align*}
By rearranging the terms, we get
\begin{align*}
\frac{4n}{\Delta(G)+2} \leq \left|S'\cup S''\right| +2\left|S'''\right|\leq \overleftrightarrow{MTT}(G,\tau).
\end{align*}

\paragraph{Tightness.} Consider the complete bipartite graph $K_{2,2\ell}$ with partite sets $\{v_1,v_2\}$ and $\{u_1,u_2,\ldots,u_{2\ell}\}$. We observe that $\mathcal{S}_0,\mathcal{S}_1,\mathcal{S}_2$ with $\mathcal{S}_0=\mathcal{S}_1=\{v_1,v_2\}$ and $\mathcal{S}_2=\emptyset$ is a TTS of $K_{2,2\ell}$ of size $4$ for the strict majority threshold assignment. Furthermore, since $\Delta(K_{2,2\ell})=2\ell=n-2$, we have the bound of $\overleftrightarrow{MTT}(K_{2,2\ell},\tau)\geq\frac{4n}{\Delta(K_{2,2\ell})+2}=\frac{4n}{(n-2)+2}=4$. Therefore, the bound is tight.

\section{Proof of Theorem~\ref{hardness-thm}}\label{appendix-hardness}

We show that there is a polynomial time reduction to the \textsc{Timed Target Set Selection} problem from the problem of finding the minimum size of a TS in the progressive model.
	
For a given pair $(H,\tau')$, construct $(G,\tau)$ as follows. For each node $v\in V(H)$ add $\left\lceil\frac{d(v)}{2}\right\rceil$ copies of the complete graph $K_2$ and connect both nodes of each copy to $v$. Set
\begin{equation*}
\tau(v)=\begin{cases}
\tau'(v) & \text{ if } v\in V(H); \\
1 & \text{ if } v\in V(G)\setminus V(H).
\end{cases}
\end{equation*}

We claim that $\overrightarrow{MT}(H,\tau')=\overleftrightarrow{MTT}(G,\tau)$. First, we show that $\overrightarrow{MT}(H,\tau')\leq\overleftrightarrow{MTT}(G,\tau)$. Let the sequence $\mathcal{S}_0,\mathcal{S}_1,\ldots,\mathcal{S}_k$ be a minimum size TTS of $(G,\tau)$ in the non-progressive model. Furthermore, assume that this TTS does not include any node from the added $K_2$ sub-graphs. We can make such an assumption because instead of making a node $u$ positive in a copy $K_2$, which is attached to a node $v$, in time step $t$, we can make node $v$ positive in step $t+1$. It is straightforward to check that $\mathcal{S}=\bigcup_{i=0}^k \mathcal{S}_i$ is a TS of $(H,\tau')$ in the progressive model. So $\overrightarrow{MT}(H,\tau')\leq |\mathcal{S}| \leq \sum_{i=0}^k |\mathcal{S}_i| = \overleftrightarrow{MTT}(G,\tau)$. 

Conversely, we prove that $\overrightarrow{MT}(H,\tau')\geq\overleftrightarrow{MTT}(G,\tau)$. Assume that $\mathcal{S}$ is a minimum size TS of $(H,\tau')$ in the progressive model and makes all nodes of $H$ positive after $k$ time steps for some integer $k$. Consider $\mathcal{S}_0,\mathcal{S}_1,\ldots,\mathcal{S}_{k+3}$ with $\mathcal{S}_0=\mathcal{S}$ and $\mathcal{S}_i=\emptyset$ for $1\leq i \leq k+3$. We claim that this sequence is a TTS of $(G,\tau)$ in the non-progressive model. In time step $1$, all added copies of $K_2$ with a neighbor in $\mathcal{S}$ become positive and stay positive forever (possibly some nodes of $\mathcal{S}$ become negative in this time step). In time step $2$, all nodes of $\mathcal{S}$ become positive and stay positive until the end of the process. According to the construction of $(G,\tau)$ from $(H,\tau')$ after $k$ time steps, i.e., time step $k+2$, all nodes of $G$ except possibly some in the attached $K_2$ subgraphs become positive. In time step $k+3$, all nodes of $G$ are positive. So $\overleftrightarrow{MTT}(G,\tau) \leq \sum_{i=0}^{k+3} |\mathcal{S}_i| = |\mathcal{S}| = \overrightarrow{MT}(H,\tau')$.

Assume that there is a polynomial time algorithm for finding $\overleftrightarrow{MTT}(G,\tau)$ with the approximation ratio $C2^{\log^{1-\epsilon}|V(G)|}$ for some constants $C,\epsilon>0$. Since $|V(G)|=\mathcal{O}(|V(H)|^2)$, the above polynomial time reduction gives us a polynomial time algorithm with approximation ratio $C'2^{\log^{1-\epsilon'}|V(H)|}$, for some constants $\epsilon',C'>0$, for the problem of finding the minimum size of a TS in the progressive model. However, this is not possible according to the results from~\cite{C}, unless $NP\subseteq DTIME(n^{polylog(n)})$.

\section{Correctness of ILP}
\label{appendix-ilp}
We need to prove that a minimum size TTS which makes all nodes positive in $k$ steps corresponds to a solution of the ILP and vice versa. For the sake of readability, we provide a constructive argument which can be used to prove this statement.

Since the goal is to find a TTS of minimum size, the objective function is set to $\sum_{v\in V}\sum_{i=0}^{k}x_{vi}$ which represents the size of a TTS. We have the constraint $x_{vk}+y_{vk}=1$ for each node $v\in V$, which ensures that all nodes become positive in $k$ steps. Furthermore, note that for all $v\in V$ and all $i\in \mathbb{K}_0$ the constraint $x_{vi}+y_{vi}\leq 1$ guarantees that no positive node $v$ in time $i$ will be in $\mathcal{S}_i$. (This is not required for its own sake since it is imposed implicitly by minimizing the objective function, but it is necessary as a combination with the next constraints, as we discuss.) 
	
The process in the non-progressive model has two rules: 
\begin{enumerate}
		\item If the number of positive neighbors of node $v$ in time step $i-1$ is greater than or equal to $\tau(v)$, then $v$ becomes positive in time step $i$. This is equivalent to: 
		\begin{equation}\label{rule1}
			\sum_{u\in N(v)}(x_{u(i-1)}+y_{u(i-1)})\geq \tau(v)\implies y_{vi}=1.
		\end{equation}
		\item If the number of positive neighbors of node $v$ in time step $i-1$ is strictly less than $\tau(v)$, then $v$ becomes negative in time step $i$ (unless we force it to be positive, i.e., $x_{vi}=1$). This is equivalent to: 
		\begin{equation}\label{rule2}
			\sum_{u\in N(v)}(x_{u(i-1)}+y_{u(i-1)})< \tau(v)\implies y_{vi}=0.
		\end{equation}
\end{enumerate}
For the first rule, we recall the following constraint: 
\begin{multline}\label{constraint1}
(d(v)+1-\tau(v))y_{vi}+\tau(v)-1\ge \\  \sum_{u\in N(v)} (x_{u(i-1)}+y_{u(i-1)}).
\end{multline}
Note that in case $y_{vi}=1$, the constraint \eqref{constraint1} reduces to $d(v)\geq\sum_{u\in N(v)} (x_{u(i-1)}+y_{u(i-1)})$ which is trivial due to the constraint $x_{vi}+y_{vi}\leq 1$ for all $v\in V$ and $i\in\mathbb{K}_0$. For $y_{vi}=0$, the constraint \eqref{constraint1} reduces to $\sum_{u\in N(v)} (x_{u(i-1)}+y_{u(i-1)})\leq \tau(v)-1$. This means that
$$y_{vi}=0\implies\sum_{u\in N(v)} (x_{u(i-1)}+y_{u(i-1)})\leq \tau(v)-1.$$
This is the contrapositive of \eqref{rule1} and so equivalent to it. So constraint \eqref{constraint1} guarantees the first rule of the process. 
	
For the second rule consider the following constraint: 
\begin{equation}\label{constraint2}
		\sum_{u\in N(v)}(x_{u(i-1)}+y_{u(i-1)})-\tau(v)y_{vi}\geq 0
\end{equation}
Note that in case $y_{vi}=0$, the constraint \eqref{constraint2} reduces to $\sum_{u\in N(v)} (x_{u(i-1)}+y_{u(i-1)})\geq0$ which is trivial. In case $y_{vi}=1$, the constraint \eqref{constraint2} reduces to $\sum_{u\in N(v)} (x_{u(i-1)}+y_{u(i-1)})\geq \tau(v)$. This means that
$$y_{vi}=1\implies\sum_{u\in N(v)} (x_{u(i-1)}+y_{u(i-1)})\geq \tau(v).$$
This is the contrapositive of \eqref{rule2} and so equivalent to it. Hence, constraint \eqref{constraint2} guarantees the second rule of the process. 

\section{Correctness of Algorithm~\ref{modified greedy}}
\label{appendix-greedy-correctness}

\begin{thm}
Algorithm~\ref{modified greedy} finds a TTS in the non-progressive model for a given pair $(G,\tau)$.
\end{thm}
\begin{proof}
The \textit{for} loop in line 3 of the algorithm determines which nodes belong to $\mathcal{S}_0$ or $\mathcal{S}_1$ or none of them. For the $i$-th iteration of this loop, let $\mathcal{S}_0^i$ and $\mathcal{S}_1^i$ denote $\mathcal{S}_0\cup\{v_{i+1},\ldots,v_n\}$ and $\mathcal{S}_1$ respectively. By induction, we show that for each $0\leq i\leq n$ the sequence $\mathcal{S}_0^i, \mathcal{S}_1^i, \emptyset$ is a TTS. For $i=0$ we have $\mathcal{S}_0^i=V$ and clearly the claim is true. Now, we show that if $\mathcal{S}_0^i,\mathcal{S}_1^i,\emptyset$ is a TTS, then $\mathcal{S}_0^{i+1},\mathcal{S}_1^{i+1},\emptyset$ is also a TTS. If $v_{i+1}\in \mathcal{S}_0^{i+1}$, since $\mathcal{S}_1^i\subseteq \mathcal{S}_1^{i+1}$ and $\mathcal{S}_0^i\subseteq \mathcal{S}_0^{i+1}$ clearly $\mathcal{S}_0^{i+1},\mathcal{S}_1^{i+1},\emptyset$ is a TTS. So let $v_{i+1}\notin \mathcal{S}_0^{i+1}$. This only happens in the two following cases:
\begin{itemize}
\item \textbf{Case 1:} $|\texttt{blocked}[v_{i+1}]|=0$. In this case, there is no node $u$ adjacent to $v_{i+1}$ such that $u$ has $d(u)-\tau(u)$ neighbors in $V(G)\setminus \mathcal{S}_0^i$. This implies that $\mathcal{S}_0^i\setminus \{v_{i+1}\},\mathcal{S}_1^i,\emptyset$ is a TTS. So $\mathcal{S}_0^{i+1},\mathcal{S}_1^{i+1},\emptyset$ is a TTS.
\item \textbf{Case 2:} $\texttt{blocked}[v_{i+1}]$ includes only $w$, for some node $w\in N(v_i)$ and $d(w)>d(v_{i+1})$. This implies that $\mathcal{S}_0^i\setminus \{v_{i+1}\},\mathcal{S}_1^i\cup\{w\},\emptyset$ is a TTS. So $\mathcal{S}_0^{i+1},\mathcal{S}_1^{i+1},\emptyset$ is a TTS. Note that since $w$ does not need the nodes of $\mathcal{S}_0$ to become positive, we can set $\tau(w)=0$.
\end{itemize} 
\end{proof}

\section{Correctness of Algorithm~\ref{Algorithm3}}
\label{appendix-tree}

We prove the correctness of Algorithm~\ref{Algorithm3}, in a constructive and step by step fashion, by providing a series of Lemmas.

\begin{lemma}\label{threshold0}
Let $G=(V,E)$ be a graph with threshold assignment $\tau$, where $\tau(v)=0$ for some node $v\in V$. Consider $G\setminus v$ with threshold assignment $\tau'$ where 
\begin{equation*}
\tau'(u)=
\begin{cases}
\tau(u)-1& u\in N(v), \tau(u)\geq 1;\\
\tau(u)&\text{otherwise}.
\end{cases}
\end{equation*}
Then, we have $\overleftrightarrow{MTT}(G,\tau)=\overleftrightarrow{MTT}(G\setminus v,\tau^{\prime})$.
\end{lemma}
\begin{proof}
Let $\mathcal{S}_0,\mathcal{S}_1,\ldots,\mathcal{S}_k$ be a TTS of minimum size for $(G,\tau)$. We claim that $v\notin \bigcup_{i=0}^{k}\mathcal{S}_i$ because otherwise $\emptyset,\mathcal{S}_0\setminus v,\mathcal{S}_1\setminus v,\ldots,\mathcal{S}_k\setminus v$ is a TTS with a smaller size which is a contradiction. Obviously, $\mathcal{S}_0,\mathcal{S}_1,\ldots,\mathcal{S}_k$ is a TTS for $(G\setminus v,\tau')$ which implies that $\overleftrightarrow{MTT}(G\setminus v,\tau^{\prime})\leq\overleftrightarrow{MTT}(G,\tau)$. On the other hand, if $\mathcal{S}_0,\mathcal{S}_1,\ldots,\mathcal{S}_k$ is a TTS for $(G\setminus v,\tau')$, then $\emptyset,\mathcal{S}_0,\mathcal{S}_1,\ldots,\mathcal{S}_k$ is a TTS for $(G,\tau)$ and hence $\overleftrightarrow{MTT}(G,\tau)\leq\overleftrightarrow{MTT}(G\setminus v,\tau^{\prime})$. (We set $\emptyset$ at the beginning of the sequence to let node $v$ become positive.) 
\end{proof}
According to Lemma~\ref{threshold0}, we can assume that there is no node whose threshold is zero. Furthermore, as mentioned before, we assume that $\tau(v)\le d(v)$ because otherwise node $v$ can not become and remain positive. This implies that for any leaf node $v$ of a tree we have $\tau(v)=1$.

\begin{lemma}\label{lemma2}
Let $T$ be a tree with threshold assignment $\tau$ and $v$ be a node belonging to set $A$. If $\mathcal{S}_0,\mathcal{S}_1,\ldots, \mathcal{S}_k$ is a TTS of $T$, then $\sum_{i=0}^k \left|\mathcal{S}_i\cap \mathcal{L}[v]\right|\geq2$.
\end{lemma}
\begin{proof}
For a TTS $\mathcal{S}_0,\mathcal{S}_1,\ldots, \mathcal{S}_k$ with minimum value of $\sum_{i=0}^k \left|\mathcal{S}_i\cap \mathcal{L}[v]\right|$, we may assume that $V(T)\setminus \mathcal{L}[v]\subseteq \mathcal{S}_i$ for all $0\leq i\leq k$. If $\sum_{i=0}^{k}\left|\mathcal{S}_i\cap \mathcal{L}[v]\right|\leq1$, then $v$ and its leaf neighbors can not get positive simultaneously in any step and this contradicts the assumption that $\mathcal{S}_0,\mathcal{S}_1,\ldots, \mathcal{S}_k$ is a TTS. 	
\end{proof}
For any two distinct nodes $u$ and $v$, we have $\mathcal{L}[u]\cap \mathcal{L}[v]=\emptyset$. Thus, it is immediate to derive Corollary~\ref{corollary1} from Lemma~\ref{lemma2}.
\begin{corollary}\label{corollary1}
For a pair $(T,\tau)$, we have $\overleftrightarrow{MTT}(T,\tau)\geq 2|A|$.
\end{corollary}
\begin{lemma}\label{A}
For a pair $(T,\tau)$, suppose that all non-leaf nodes of $T$ (except possibly one) belong to $A'\cup B$ (i.e., $|A^{\prime\prime}\cup C|\le 1$). Then, we have $\overleftrightarrow{MTT}(T,\tau)=2|A|$. (There is in fact a minimum size TTS for which every node that becomes positive will remain positive until the end of the process.)
\end{lemma}
\begin{proof}
Let $v$ be a non-leaf node, which belongs to $A''\cup C$ in case $A''\cup C\neq \emptyset$. Consider the rooted version of $T$ with root $v$. We prove by induction on the depth of tree $d$. 
	
For $d=1$, $T$ is a star with $v$ in the center. So $\bar{l}(v)=0$. Since we assumed there is no node with threshold $0$, $\tau(v)>\bar{l}(v)$ and hence $v\in A$. Clearly, $\mathcal{S}_0=\{v\},\mathcal{S}_1=\{v\}$ is a TTS of size $2$. Using Lemma~\ref{lemma2}, we conclude that $\overleftrightarrow{MTT}(T,\tau)=2$. So we have $\overleftrightarrow{MTT}(T,\tau)=2|A|$. Clearly, in this TTS every positive node does not become negative in any step. 
	
Suppose that the statement is true for all trees with $d<r$. Let $T$ be a tree with a non-leaf node $v$ as its root such that $d=r$ and for every $u\in V(T)\setminus \{v\}$ we have $u\notin A''\cup C$. Let $u_1,\ldots,u_{\bar{l}(v)}$ be non-leaf children of $v$. For each $1\leq i \leq \bar{l}(v)$, let $T_i$ denote the induced subtree on $u_i$ and its descendants. Denote by $\tau_i$ the restriction of $\tau$ to $V(T_i)$ as the threshold assignment of $T_i$. We claim that $(T_i,\tau_i)$ rooted at $u_i$ has the conditions of the induction hypothesis.
	
First, note that $\tau_i(u_i)\leq d_{T_i}(u_i)$ because otherwise $\tau(u_i) = \tau_i(u_i) > d_{T_i}(u_i) = d(u_i)-1$ which implies $\tau(u_i)=d(u_i)$. If $u_i$ does not have leaf neighbors, then $\tau(u_i)=d(u_i)=\bar{l}(u_i)$ and so $u_i\in C$. In case $u_i$ has a leaf neighbor, we have $\tau(u_i)=d(u_i)>\bar{l}(u_i)$ and so $u_i\in A''$. Both cases contradict the fact that the only node of $T$ belonging to $A''\cup C$ can be $v$. 
	
The nodes of $T_i$ other than $u_i$, have the same threshold and degree as theirs in $T$ and hence they are not in $A_i''\cup C_i$. So $T_i$ has our desired conditions and by the induction hypothesis $\overleftrightarrow{MTT}(T_i,\tau_i)=2|A_i|$ and it has a minimum size TTS for which any positive node never becomes negative. Note that $A_i,B_i,C_i,A'_i$ and $A''_i$ are defined similar to $A,B,C,A'$ and $A''$ respectively for subtree $T_i$.
	
Since $u_i\notin C$, we consider the following two cases for $u_i$, where $\bar{l}_{T_i}(u_i)= |\bar{\mathcal{L}}(u_i) \cap T_i|$:
	\begin{enumerate}
		\item $u_i\in A$, then $\bar{l}_{T_i}(u_i)=\bar{l}(u_i)-1<\tau(u_i)-1<\tau(u_i) = \tau_i(u_i)$ and so we have $u_i\in A_i$. 
		\item $u_i\in B$, then $\tau(u_i)\leq \bar{l}(u_i)-1=\bar{l}_{T_i}(u_i)$ and so in the case $\tau(u_i)<\bar{l}_{T_i}(u_i)$, $u_i\in B_i$ and in the case $\tau(u_i)=\bar{l}_{T_i}(u_i)$, $u_i\in C_i$. 
	\end{enumerate}
According to the above cases for $u_i$ and what mentioned about the degree and threshold of the other nodes of $T_i$ we have: 
	\begin{equation}\label{key1}
		\left|A\right|=\sum_{i=1}^{\bar{l}(v)}\left|A_i\right|+f(v)
	\end{equation}
where $f(v)=1$ in case $v\in A$ and $f(v)=0$ otherwise. 
	
For each $1\leq i\leq\bar{l}(v)$, suppose that $\mathcal{S}_{i,0},\mathcal{S}_{i,1},\ldots,\mathcal{S}_{i,k_i}$ is a TTS of $T_i$ of minimum size where every positive node remains positive until the end of the process. Set $k:=\max \{k_i:1\leq i \leq \bar{l}(v)\}$. For each $k_i < j \leq k$ set $\mathcal{S}_{i,j}:=\emptyset$. For each $1\leq j\leq k$, set $\mathcal{S}_j:=\bigcup_{i=1}^{\bar{l}(v)}\mathcal{S}_{i,j}$. 
	
If $v\in B\cup C$, then $\mathcal{S}_0,\mathcal{S}_1,\ldots,\mathcal{S}_k,\mathcal{S}_{k+1}=\emptyset,\mathcal{S}_{k+2}=\emptyset$ will be a TTS of $T$. This is because all nodes in $T_i$'s become positive until the $k$-th step and will stay positive. In the $(k+1)$-th step, $v$ and in the $(k+2)$-th step the leaves of $v$ become positive and they will remain positive. 
	
If $v\in A$, then $\mathcal{S}_0,\mathcal{S}_1,\ldots,\mathcal{S}_k,\mathcal{S}_{k+1}=\{v\},\mathcal{S}_{k+2}=\{v\}, \mathcal{S}_{k+3}=\emptyset$ will be a TTS. This is because all nodes in $T_i$'s become positive until the $k$-th step and will remain positive. In the $(k+1)$-th step, $v$ and in the $(k+2)$-th step, the leaves of $v$ become positive and remain positive.
	
According to Equation~\eqref{key1}, the timed target sets presented in both cases give $\overleftrightarrow{MTT}(T,\tau)\leq2|A|$. This inequality together with Lemma~\ref{lemma2} completes the proof. 
\end{proof}
\begin{remark}\label{remA}
Note that the proof of the previous lemma is constructive, and it gives an algorithm for finding such a minimum size TTS of $(T,\tau)$. 
\end{remark}

	
	

\begin{lemma}\label{C}
Consider a pair $(T,\tau)$, where $T$ is rooted at a node $v\in A''\cup C$. Suppose that $u\in V(T)\setminus \{v\}$ is a node belonging to $C$ such that none of its descendants belongs to $A''\cup C$. Let $T'$ be the induced subtree on $u$ and its descendants and $T''$ be the induced subtree on $\left(V(T)\setminus V(T')\right)\cup\{u\}$. Let also $\tau'$ and $\tau''$ be the threshold assignment of $T'$ and $T''$ respectively as follows: 
\begin{equation*}
\tau'(w)=
\begin{cases}
\tau(w) & w\in V(T')\setminus\{u\} \\
\tau(w)-1 & w=u
\end{cases}
\end{equation*}
\begin{equation*}
\tau''(w)=
\begin{cases}
\tau(w) & w\in V(T'')\setminus\{u\} \\
1 & w=u.
\end{cases}
\end{equation*}
Then, we have $\overleftrightarrow{MTT}(T,\tau)=\overleftrightarrow{MTT}(T',\tau^{\prime})+\overleftrightarrow{MTT}(T'',\tau^{\prime\prime})$.
\end{lemma}
\begin{proof}
Assume that $\mathcal{S}_0',\mathcal{S}_1',\ldots,\mathcal{S}_{k'}'$ is a minimum size TTS for $(T',\tau')$, following the construction from Lemma~\ref{A}. Note that due to the construction method, we may suppose that $\mathcal{S}_{k'-1}'=\mathcal{S}_{k'}'=\emptyset$, $u\in Q_{k'-1}'$ and $\mathcal{L}(u)\subseteq Q_{k'}'$. Furthermore, let $\mathcal{S}_0'',\mathcal{S}_1'',\ldots,\mathcal{S}_{k''}''$ be a TTS of minimum size for $(T'',\tau'')$. Then clearly $\mathcal{S}_0',\mathcal{S}_1',\ldots,\mathcal{S}_{k'-2}',\mathcal{S}_0'',\mathcal{S}_1'',\ldots,\mathcal{S}_{k''}''$ is a TTS for $(T,\tau)$. Note that in this TTS of $T$ the node $u$ will become positive after its parent becomes positive and one step later its leaf neighbors become positive. Hence 
\begin{align*}
\overleftrightarrow{MTT}(T,\tau) \leq & \sum_{i=0}^{k'}|\mathcal{S}_i'|+\sum_{i=0}^{k''}|\mathcal{S}_i''|\\ = &\overleftrightarrow{MTT}(T',\tau^{\prime})+\overleftrightarrow{MTT}(T'', \tau^{\prime\prime}).
\end{align*}
	
Now, we will show that $\overleftrightarrow{MTT}(T,\tau)\geq\overleftrightarrow{MTT}(T',\tau^{\prime})+\overleftrightarrow{MTT}(T'',\tau^{\prime\prime})$. Let $\mathcal{S}_0,\mathcal{S}_1,\ldots,\mathcal{S}_k$ be a minimum size TTS of $T$. Let the sets $A_{T'},B_{T'}$ and $C_{T'}$ be defined for $T'$ similar to $A,B$ and $C$ for tree $T$. Since $u\in C$ for tree $T$, it is easy to see that $u\in C_{T'}$ for $T'$. Hence, using Lemma~\ref{lemma2} yields
	\begin{equation*}
		\sum_{i=0}^{k}\left|\mathcal{S}_i\cap\left(T'\setminus \mathcal{L}[u]\right)\right| \geq 2 |A_{T'}|.
	\end{equation*}
Then, applying Lemma~\ref{A} gives
	\begin{equation}\label{key2}
		\sum_{i=0}^{k}\left|\mathcal{S}_i\cap\left(T'\setminus \mathcal{L}[u]\right)\right|\geq{\rm \overleftrightarrow{MTT}(T',\tau^{\prime})}.
	\end{equation}
Suppose that $\mathcal{S}_0',\mathcal{S}_1',\ldots,\mathcal{S}_{k'}'$ is a minimum size TTS of $T'$ such that every positive node remains positive until the end of the process, using the construction given in Lemma~\ref{A}. Then 
\begin{align*}
\mathcal{S}_0',\mathcal{S}_1',\ldots,\mathcal{S}_{k'}',\emptyset,\emptyset,
\mathcal{S}_0\setminus\left(T'\setminus \mathcal{L}[u]\right), \mathcal{S}_1\setminus\left(T'\setminus \mathcal{L}[u]\right), \\ \ldots, \mathcal{S}_k\setminus\left(T'\setminus \mathcal{L}[u]\right)
\end{align*}
is a TTS of $T$.
Thus, Equation~\eqref{key2} holds in equality, which implies that
\begin{equation}
\sum_{i=0}^{k}\left|\mathcal{S}_i\cap\left(T'\setminus \mathcal{L}[u]\right)\right|=\overleftrightarrow{MTT}(T',\tau^{\prime}).
\end{equation}
Now, we consider the following cases for $\mathcal{S}_i\cap \mathcal{L}[u]$: 
\begin{enumerate}
\item $\mathcal{S}_i\cap \mathcal{L}[u]=\emptyset$, for all $1\leq i\leq k$. 
\item $u\in \mathcal{S}_i$, for some $1\leq i\leq k$. 
\item $x\in \mathcal{S}_i\cap \mathcal{L}(u)$, for some $1\leq i\leq k$. 
\end{enumerate}
In case 1, clearly $\mathcal{S}_0\setminus T',\mathcal{S}_1\setminus T',\ldots,\mathcal{S}_k\setminus T'$ is a TTS for $T''$. So the desired inequality $\overleftrightarrow{MTT}(T,\tau)\geq\overleftrightarrow{MTT}(T',\tau^{\prime})+\overleftrightarrow{MTT}(T'',\tau^{\prime\prime})$ holds. 

In case 2, remove $u$ from $\mathcal{S}_i$ and add its parent $z$ to $\mathcal{S}_{i-1}$ and in case 3 remove $x$ from $\mathcal{S}_i$ and add $z$ to $\mathcal{S}_{i-2}$. By doing so for all nodes of cases 2 and 3, we get a new sequence $\mathcal{S}_0,\mathcal{S}_1,\ldots,\mathcal{S}_k$ which is not necessarily a TTS of $T$ (where we are abusing the notation). However, the sequence 
\begin{align*}
\mathcal{S}_0',\mathcal{S}_1',\ldots,\mathcal{S}_{k'}',\emptyset,\emptyset,
\mathcal{S}_0\setminus\left(T'\setminus \mathcal{L}[u]\right), \mathcal{S}_1\setminus\left(T'\setminus \mathcal{L}[u]\right), \\ \ldots, \mathcal{S}_k\setminus\left(T'\setminus \mathcal{L}[u]\right)
\end{align*}
constructed from it will be a TTS for $T$. Clearly, its size is not larger, and also we have $\mathcal{S}_i\cap \mathcal{L}[u]=\emptyset$ for all $0\leq i\leq k$. Applying an argument similar to the one from case 1 finishes the proof.
\end{proof}
\begin{remark}\label{remark}
Let $T$ be a tree with at most one node belonging to $A''$ (with any number of nodes belonging to $C$) and $\tau$ be a threshold assignment. Repeatedly using Lemmas~\ref{C} and~\ref{A} gives an algorithm which finds a minimum size TTS of $(T,\tau)$, such that any positive node does not become negative in any step. 
\end{remark}
\begin{lemma}\label{L9}
Consider a pair $(T,\tau)$, where $T$ is rooted at a node $v\in A''\cup C$. Suppose that $u\in V(T)\setminus \{v\}$ is a node belonging to $A''$ such that non of its descendants belongs to $A''\cup C$. Denote by $z$ the parent of $u$. Let $T'$ be the induced subtree on $u$ and its descendants and $T''$ be the induced subtree on $\left(V(T)\setminus V(T')\right)$. Let also $\tau'$ and $\tau''$ be the threshold assignment of $T'$ and $T''$ respectively as follows: 
\begin{equation*}
		\tau'(w)=
		\begin{cases}
			\tau(w) & w\in V(T')\setminus\{u\} \\
			\tau(w)-1 & w=u
		\end{cases}
\end{equation*}
\begin{equation*}
		\tau''(w)=
		\begin{cases}
			\tau(w) & w\in V(T'')\setminus\{z\} \\
			\tau(w)-1 & w=z.
		\end{cases}
\end{equation*}
Then, $\overleftrightarrow{MTT}(T, \tau)= \overleftrightarrow{MTT}(T',\tau^{\prime})+\overleftrightarrow{MTT}(T'',\tau^{\prime\prime})$ and there exists a minimum size TTS for which each positive node never becomes negative in any step.
\end{lemma}
\begin{proof}
We use induction on the size of $A''$. For $|A''|=1$, the proposition has been proved in Remark~\ref{remark}. Now, suppose that the statement is true for $|A''|<r$ and $T$ is a tree with $|A''|=r$. Consider a node $v\in A''$ as the root of $T$. Let $u$ be a node in $A''$ whose descendants do not belong to $A''\cup C$ (if there is not such a node, then using Lemma~\ref{C} we choose a descendant belonging to $C$ and remove $T'$ from $T$ and repeat this until finding such a node $u$). By the induction hypothesis, assume that $\mathcal{S}_0',\mathcal{S}_1',\ldots,\mathcal{S}_{k'}'$ and $\mathcal{S}_0'',\mathcal{S}_1'',\ldots,\mathcal{S}_{k''}''$ are minimum size timed target sets for $(T',\tau')$ and $(T'',\tau'')$, respectively, such that in both of them every positive node never becomes negative. Considering $u$ as the root of $T'$ since $u\in A''\subseteq A$ according to the proof of Lemma~\ref{A} we may assume that $\mathcal{S}_{k'-2}'=\mathcal{S}_{k'-1}'=\{u\}$ and $\mathcal{S}^{\prime}_{k'}=\emptyset$. Also note that $\mathcal{S}_0',\mathcal{S}_1',\ldots,\mathcal{S}_{k'-2}',\emptyset,\ldots,\emptyset,\{u\},\{u\}$ for any number of empty sets is a TTS of $T'$ too. 
	
Suppose that in $T''$, node $z$ becomes positive in step $i$. So by the assumption, it remains positive until the end of the process. In other words, $i$ is the smallest number such that $z\in \mathcal{S}_j''\cup Q_j''$, for all $i\leq j\leq k$. Obviously
\begin{align*}
\mathcal{S}_0',\mathcal{S}_1',\ldots,\mathcal{S}_{k'-3}',\mathcal{S}_0'',\mathcal{S}_1'',\ldots,\mathcal{S}_{i-2}'',\mathcal{S}_{i-1}''\cup\mathcal{S}_{k'-2}',\\ \mathcal{S}_{i}''\cup\mathcal{S}_{k'-1}',\mathcal{S}_{i+1}''\cup\mathcal{S}_{k'}',\mathcal{S}_{i+2}'',\ldots,\mathcal{S}_{k''}''
\end{align*}
is a TTS for $(T,\tau)$ such that every positive node will remain positive until the end of the process. This also shows that 
$$\overleftrightarrow{MTT}(T,\tau)\leq \overleftrightarrow{MTT}(T',\tau^{\prime})+ \overleftrightarrow{MTT}(T'',\tau^{\prime\prime}).$$
	
Now, we only need to prove the following inequality to complete the proof:
\begin{equation}\label{eq}
\overleftrightarrow{MTT}(T,\tau^{\prime})\geq \overleftrightarrow{MTT}(T',\tau^{\prime})+\overleftrightarrow{MTT}(T'',\tau^{\prime\prime}).
\end{equation}
Let $\mathcal{S}_0,\mathcal{S}_1,\ldots,\mathcal{S}_k$ be a minimum size TTS of $(T,\tau)$. Since $\mathcal{S}_0\cap V(T'),\mathcal{S}_1\cap V(T'),\ldots,\mathcal{S}_k\cap V(T')$ is a TTS for $(T',\tau')$, by Corollary~\ref{corollary1} we have 
\begin{equation}\label{eq1}
\sum_{i=0}^{k}|\mathcal{S}_i\cap V(T')|\geq 2|A_{T'}|=\overleftrightarrow{MTT}(T',\tau^{\prime})
\end{equation}
where the last equality is achieved by Lemma~\ref{A}. Also, we claim that 
\begin{equation}\label{eq2}
		\sum_{i=0}^{k}|\mathcal{S}_i\cap V(T'')|\geq \overleftrightarrow{MTT}(T'',\tau^{\prime\prime}).
\end{equation}
To show this, we may assume that $V(T')\subseteq \mathcal{S}_i$ for all $0\leq i\leq k$. Clearly, $N(V(T'))\cap V(T'')=\{z\}$ and the only neighbor of $z$ in $T'$ is $u$. Since $u$ is positive in any step and $\tau''(z)=\tau(z)-1$, we conclude that 
$$\mathcal{S}_0\cap V(T''),\mathcal{S}_1\cap V(T''),\ldots,\mathcal{S}_k\cap V(T'')$$
is a TTS of $T''$. This yields Equation~\eqref{eq2}. From Equations~\eqref{eq1} and~\eqref{eq2}, we get
$$\sum_{i=0}^{k}|\mathcal{S}_i|\geq \overleftrightarrow{MTT}(T',\tau^{\prime})+\overleftrightarrow{MTT}(T'',\tau^{\prime\prime})$$
which proves Equation~\eqref{eq} as desired.
	
	
	
\end{proof}
Now, we are ready to prove the correctness of Algorithm~\ref{Algorithm3} in Theorem~\ref{tree-algorithm-proof}.

	
	
	

\begin{thm}\label{tree-algorithm-proof}
For any tree $T$ and threshold assignment $\tau$, Algorithm~\ref{Algorithm3} finds $\overleftrightarrow{MTT}(T,\tau)$ in linear time. 
\end{thm}
\begin{proof}
In case of $|A''\cup C|\leq 1$, using Lemma~\ref{A} implies that the minimum size of a TTS is equal to $2|A|$, which can be calculated in linear time.
	
In Algorithm~\ref{Algorithm3}, we define a variable $x$ and set it to zero at the beginning which at the end of the algorithm will be the size of the minimum timed target set of $(T,\tau)$. In case of $|A''\cup C|\ge 2$, let $T$ be rooted at $v\in A''\cup C$. Suppose that for any node $w$ its distance from $v$ is denoted by $d(w, v)$. Find a node $u\in A''\cup C$ or a node $u$ with $\tau(u)=0$ for which $d(w,u)$ is maximum. So $u$ does not have any descendants belonging to $A''\cup C$. Now consider three cases:
\begin{enumerate}
		\item If $u\in A''$, then by defining $(T',\tau')$ and $(T'',\tau'')$ same as Lemma~\ref{L9} we have $$\overleftrightarrow{MTT}(T,\tau)=\overleftrightarrow{MTT}(T',\tau^{\prime})+\overleftrightarrow{MTT}(T'',\tau^{\prime\prime}).$$ Obviously $T'$ has the conditions of Lemma~\ref{A} and so $\overleftrightarrow{MTT}(T',\tau^{\prime})=2|A_{T'}|$. Set $x=x+2|A_{T'}|$ and $(T,\tau)=(T'',\tau'')$. This explains the first part of the \textit{if} statement in line 11 of Algorithm \ref{Algorithm3}.
  		\item If $\tau(u)=0$, then by Lemma~\ref{threshold0} if we remove $u$ from the graph and for all of its neighbors we reduce the threshold by one, then the size of a TTS does not change. Obviously, the minimum size of a TTS of $T\setminus\{u\}$ is the summation of it over the components. Denote by $T''$ the component of $T\setminus \{u\}$ including node $z$ (the parent of $u$) and suppose that $\tau''$ is the restriction of $\tau$ to $V(T'')\setminus \{z\}$ and $\tau''(z)=\tau(z)-1$. Let $T'=T\setminus T''$ and suppose that $\tau'$ is the restriction of $\tau$ over $V(T')$. Components of $T'\setminus \{u\}$ are the components of $T\setminus \{u\}$ other than $T''$. Using Lemma~\ref{threshold0} for $T'$ we conclude that   
		$$\overleftrightarrow{MTT}(T,\tau)=\overleftrightarrow{MTT}(T',\tau^{\prime})+\overleftrightarrow{MTT}(T'',\tau^{\prime\prime}).$$
		Consider $u$ as the root of $T'$ and also consider $u$ as a non-leaf node even if it is really a leaf in $T'$. It is easy to see that, although $T'$ has a node with threshold zero, we can use Lemma~\ref{A} for $T'$ and so $\overleftrightarrow{MTT}(T',\tau^{\prime})=2|V(T') \cap A|$. Set $x=x+2|V(T') \cap A|$ and $(T,\tau)=(T'',\tau'')$. This is equivalent to the second part of the \textit{if} statement in line 11 of Algorithm \ref{Algorithm3}.
  		\item If $u\in C$, then by defining $(T',\tau')$ and $(T'',\tau'')$ same as Lemma~\ref{C} we have $$\overleftrightarrow{MTT}(T,\tau)=\overleftrightarrow{MTT}(T',\tau^{\prime})+\overleftrightarrow{MTT}(T'',\tau^{\prime\prime}).$$ Obviously $T'$ has the conditions of Lemma~\ref{A} and so $\overleftrightarrow{MTT}(T',\tau^{\prime})=2|A_{T'}|$. 
		Set $x=x+2|A_{T'}|$ and $(T,\tau)=(T'',\tau'')$. This is equivalent to line 14 of Algorithm \ref{Algorithm3}.

We observe that when we loop over the nodes from $L_d$ to $L_1$, we always encounter one of the above cases. At the end, we are left with a subtree with no node in $A^{\prime\prime}\cup C$ (except potentially, the root). That is why we set $x+=2|V(T)\cap A|$ at the end. Furthermore, it is straightforward to check that the algorithm runs in linear time in the number of edges and nodes in the tree, which are both in $\Theta(n)$.
\end{enumerate}
\end{proof}

\end{document}